%% file: main.tex
\documentclass{amsart}
\usepackage[T1]{fontenc}

\newif\ifcameraready

\newcommand{\switch}[1]{\includegraphics{#1}}

\camerareadyfalse

\usepackage[margin=1.5in]{geometry}
\usepackage{amsmath, amssymb, amsthm, mathtools}
\usepackage{amsaddr}
\usepackage{booktabs}
\usepackage{standalone}
\usepackage{orcidlink}

\usepackage[p,osf]{cochineal}
\usepackage[cochineal]{newtxmath}

\usepackage{graphicx}
\usepackage{hyperref}
\hypersetup{
    colorlinks=true,
    linkcolor=blue,
    citecolor=blue,     
    urlcolor=blue
}
\usepackage{subcaption}

\usepackage{microtype}

\newtheorem{theorem}{Theorem}
\newtheorem{lemma}[theorem]{Lemma}

\newtheorem{definition}[theorem]{Definition}
\usepackage{color}

\usepackage{cleveref}

\newcommand{\sqa}{s_1}
\newcommand{\sqb}{s_2}
\newcommand{\sqc}{s_3}
\newcommand{\sqd}{s_4}
\newcommand{\boxb}[1]{#1'}

\begin{document}

\makeatletter
\newcommand{\stx}[1]{\mathpalette\giusti@stx{#1}}
\newcommand{\giusti@stx}[2]{%
  \mbox{%
    \medmuskip=\thinmuskip
    \thickmuskip=\thinmuskip
    $\m@th#1#2$%
  }%
}
\makeatother
\newcommand{\mytimes}{\mathbin{\mkern-2mu\times\mkern-2mu}}

\newcommand{\boxDims}[3]{\stx{#1 {\times} #2 {\times} #3}}

\newcommand{\minPossible}{\Psi}
\newcommand{\minImpossible}{\Delta}

\title{Unfolding Boxes with Local Constraints}
%
%
\author{Long Qian \orcidlink{0000-0003-1567-3948},
Eric Wang \orcidlink{0009-0004-6636-4541}, Bernardo Subercaseaux \orcidlink{0000-0003-2295-1299}
\and Marijn J. H. Heule \orcidlink{0000-0002-5587-8801} }

\email{\{longq,ebwang,bsuberca,mheule\}@andrew.cmu.edu}
\address{Carnegie Mellon University, Pittsburgh, PA 15213, USA}
%
%
%
%
\begin{abstract}
We consider the problem of finding and enumerating polyominos that can be folded into multiple non-isomorphic boxes.
While several computational approaches have been proposed, including SAT, randomized algorithms, and decision diagrams, none has been able to perform at scale.
We argue that existing SAT encodings are hindered by the presence of global constraints (e.g., graph connectivity or acyclicity), which are generally hard to encode effectively and hard for solvers to reason about.
In this work, we propose a new SAT-based approach that replaces these global constraints with simple local constraints that have substantially better propagation properties. 
Our approach dramatically improves the scalability of both computing and enumerating common box unfoldings: (i) while previous approaches could only find common unfoldings of two boxes up to area 88, ours easily scales beyond 150, and (ii) while previous approaches were only able to enumerate common unfoldings up to area 30, ours scales up to 60. This allows us to rule out 46, 54, and 58 as the smallest areas allowing a common unfolding of three boxes, thereby refuting a conjecture of Xu et al.\ (2017).
\end{abstract}
\vspace*{-1.5cm}
\maketitle
\section{Introduction}\label{sec:intro}
\input{sections/intro.tex}

\section{Overview of our approach}\label{sec:overview}
\input{sections/overview.tex}

\section{Local constraints}\label{sec:local}
\input{sections/local_constraints.tex}

\section{Common unfoldings of two boxes}\label{sec:common_unfoldings}
\input{sections/common_unfoldings.tex}

\section{Symmetry breaking}\label{sec:symmetry}
\input{sections/symmetry_breaking.tex}

\section{Experimental results}\label{sec:experiments}
\input{sections/experimental.tex}

\section{Related work}\label{sec:related}
\input{sections/related_work.tex}

\section{Concluding remarks}\label{sec:conclusions}
\input{sections/conclusion.tex}

%
%
%

\newpage

\bibliographystyle{splncs04}
\bibliography{bibliography}

\ifcameraready
\else

\newpage

\appendix

\section{Proofs for local constraints}\label{sec:formalization}
\input{sections/formalization.tex}

\section{Additional Unfoldings}\label{sec:extra_unfolding}
\input{sections/add_unfold.tex}

\fi

\end{document}

%% file: sections/intro.tex
Folding two-dimensional surfaces into a three-dimensional structure is a fundamental problem in computational geometry, with many applications ranging from the arts (e.g., origami) to diverse fields of engineering (e.g., packaging, protein folding)~\cite{demaine2007geometric}.
Perhaps the simplest example, usually introduced to children, is that of folding a $\boxDims{1}{1}{1}$ box (\Cref{fig:111box_intro_b}) from a \emph{net} corresponding to a polyomino of area 6 (\Cref{fig:111box_intro_a}). Interestingly, as depicted in~\Cref{fig:111box_intro_c}, multiple nets can fold into the same box, and the problem of enumerating all the nets that fold into a given $\boxDims{a}{b}{c}$ box is already non-trivial (see~\Cref{tab:11n_nets}).

\begin{figure}[h]
    \begin{subfigure}{0.3\textwidth}
        \centering
        \switch{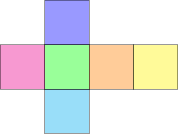}
        \caption{Standard net.}\label{fig:111box_intro_a}
    \end{subfigure}
    \begin{subfigure}{0.3\textwidth}
       \centering 
        \switch{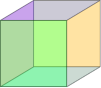}
        \caption{$\boxDims{1}{1}{1}$ box.}\label{fig:111box_intro_b}
    \end{subfigure}
    \begin{subfigure}{0.3\textwidth}
        \centering
        \switch{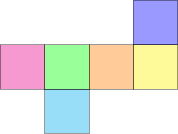}
        \caption{Non-standard net.}\label{fig:111box_intro_c}
    \end{subfigure}
    \caption{Illustration of the non-uniqueness of nets that fold into a box.}
    \label{fig:111box_intro}
    \end{figure}

An even more surprising fact, depicted in~\Cref{fig:non_isomorphic_boxes}, is that sometimes multiple non-isomorphic boxes can be obtained from the same net, simply by folding along different edges.
These~\emph{common unfoldings} (also known as~\emph{common developments}~\cite{XU20171}) allow for interesting engineering applications: a single two-dimensional piece of cardboard can be used for different types of boxes, depending on the dimensions of the object to be packed.

\begin{figure}[h]
    \centering
    \switch{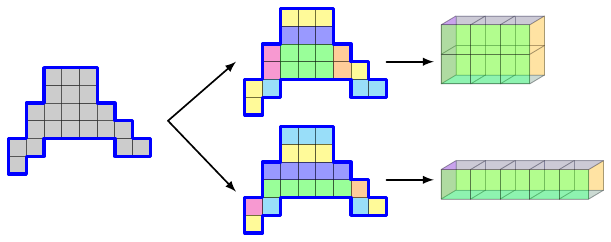}
    \caption{Two non-isomorphic boxes that can be folded from the same net.}\label{fig:non_isomorphic_boxes}
\end{figure}

A body of work has focused in the particular case in which the nets are \emph{polyominos}, and one can only fold along the edges of the unit squares forming the polyomino (see~\Cref{fig:non_isomorphic_boxes})~\cite{randomAlg2008,threeBoxes2012,tadaki2020searchdevelopmentsboxhaving,ueharaSurvey2015,XU20171}. Nonetheless, several important questions remain open. To state them, let us define some notation. 

For a positive integer $s$ representing a surface area, let $P(s)$ be the set of all triples $(a, b, c)$ with $a \leq b \leq c$ such that $s = 2(ab + ac + bc)$. In other words, $P(s)$ is the set of all possible integer dimensions of a box with surface area $s$. For example, $P(22) = \{(1,1,5), (1, 2, 3)\}$, and as shown in~\Cref{fig:non_isomorphic_boxes}, it turns out that both boxes in $P(22)$ can be folded from the same net of area $22$. 
For a positive integer $k$, let $\minPossible(k)$ be the smallest integer $s$ such that there is a subset $P' \subseteq P(s)$ with $|P'| = k$ and all boxes in $P'$ have a common unfolding. 
The example in~\Cref{fig:non_isomorphic_boxes} shows that $\minPossible(2) \leq 22$, and it can be easily checked that no smaller value of $s$ works. An impressive result by Shirakawa and Uehara is that $\minPossible(3) \leq 532$, as they showed a common unfolding for boxes $\boxDims{7}{8}{14}$, $\boxDims{2}{4}{43}$, and $\boxDims{2}{13}{16}$. Xu et al.~\cite{XU20171} conjectured $\minPossible(3) = 46$, which is the smallest value for which $|P(s)| \geq 3$.
For $k \geq 4$, it is not even known whether $\minPossible(k)$ is finite.
We define as well the opposite quantity, $\minImpossible(k)$, as the smallest integer $s$ such that there is a subset $P' \subseteq P(s)$ with $|P'| = k$ where \textbf{no} common unfolding for $P'$ exists. The work of Mitani and Uehara~\cite{randomAlg2008} shows that $\minImpossible(2) > 38$, and suggests $\minImpossible(2) > 88$. On the other hand, no upper bounds for $\minImpossible(2)$ are known. 

\begin{figure}[h]
    \begin{subfigure}{0.49\textwidth}
    \centering
        \switch{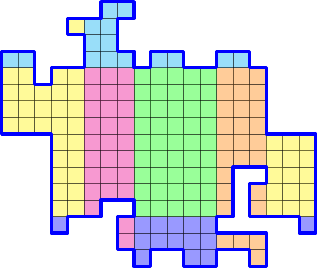}
        \caption{$\boxDims{3}{5}{9}$}
    \end{subfigure}    
    \begin{subfigure}{0.5\textwidth}
        \centering
        \switch{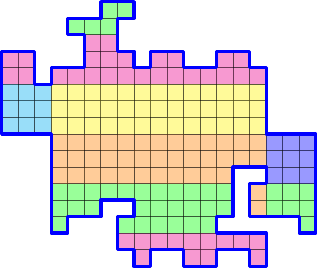}
        \caption{$\boxDims{3}{3}{13}$}
    \end{subfigure}
    \caption{A common unfolding of area 174.}\label{fig:area_174}   
\end{figure}

In this work, we show that $\minPossible(3) > 58$, that $\minImpossible(3) = 46$. We show as well that $\minImpossible(2) > 86$, and using heuristics we manage to compute solutions for certain very large areas (see~\Cref{fig:area_174}). More importantly, the SAT-based approach we developed to prove these results showcases a more general idea that might be applicable in a wide variety of contexts: global constraints, for which conflicts can be detected only after long propagation chains, can be approximated by local constraints for which conflicts are detected much earlier, leading to substantially better performance.
To illustrate the power of our approach, let us present some brief elements in comparison with previous research. 
In 2011, Abel et al.~\cite{22allsols2011} enumerated all common unfoldings for boxes $\boxDims{1}{1}{5}$ and $\boxDims{1}{2}{3}$ in about $10$ hours~\cite{XU20171}, whereas our approach allows a complete enumeration in $2$ minutes on a personal computer. Xu et al.~\cite{XU20171} remarked that using the same approach for area $30$ would have taken \emph{``too huge memory even on a supercomputer''}, and their more efficient ZDD-based approach took only $10$ days for area $30$; ours takes $10$ minutes. Moreover, Xu et al. conjectured that $\minPossible(3) = 46$, saying \emph{``However, the number of polygons of area 46 seems to be too huge to search''}. On a supercomputer~\cite{brownBridges2PlatformRapidlyEvolving2021}, our approach took $3$ hours to refute this conjecture.

\bigskip

\paragraph{\bf Code.} Our code and the instructions to reproduce our results are publicly available at\\ \url{https://github.com/LongQianQL/CADE30-BoxUnfoldings}.

\bigskip

\paragraph{\bf Organization.} In~\Cref{sec:overview}, we give a high-level overview of our approach. Then,~\Cref{sec:valid_cut_edges} introduces necessary properties that box unfoldings must satisfy, and that guide our encoding. In~\Cref{sec:local}, we detail the differences with previous SAT encodings for a spanning tree (a natural subproblem for finding unfoldings) as we use local constraints that approximate both connectivity and acyclicty. Then, in~\Cref{sec:common_unfoldings}, we show how to encode that a folding net, which our encoding keeps implicit, maps to two different boxes. In~\Cref{sec:symmetry} we show how we break rotational symmetries of the problem. \Cref{sec:related} discusses related work, and in particular, the limitations of previous SAT encodings. In~\Cref{sec:experiments} we present our experimental results. 

%% file: sections/overview.tex
\begin{figure}[t]
    \begin{subfigure}{0.49\textwidth}
      \centering
      \switch{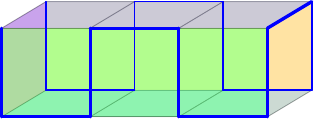}
    \caption{A $1 \times 1 \times 3$ box with cut edges.}\label{fig:113_box}
    \end{subfigure}
    \hfill
    \begin{subfigure}{0.49\textwidth}
        \centering
      \switch{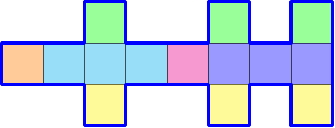}
      \caption{The corresponding folding net.}\label{fig:113_net}
    \end{subfigure}
    \caption{The correspondence between cuts in a box and its folding nets.}\label{fig:boxnet}
  \end{figure}

  The first step toward finding common unfoldings of multiple boxes is finding unfoldings of a single box. As in previous work (e.g.,~\cite{demaine2007geometric,tadaki2020searchdevelopmentsboxhaving,XU20171}), we frame the search for unfoldings of a given box $B$ in terms of a search for \emph{``cut edges''} in $B$ (see~\Cref{fig:boxnet}) which after being physically cut, would allow to unfold the box into a flat net without overlaps. 
  For instance, if one were to cut the blue edges of the $\boxDims{1}{1}{3}$ box depicted in~\Cref{fig:113_box}, and then proceed to unfold\footnote{A precise mathematical definition of folding/unfolding turns out to be pretty intricate, using the entirety of Chapter 11 in the book of Demaine and O'Rourke~\cite{demaine2007geometric}. We will thus mostly stick to intuition.} the box, the result would be the net depicted in~\Cref{fig:113_net}.
  More precisely, any net $N$ that folds into a box $B$ through a sequence $\gamma_1, \ldots, \gamma_k$ of folding motions in space can be obtained by cutting a subset of the edges of the unit-squares that compose the different faces of $B$, and then reversing the folding motions as $\gamma_k^{-1}, \ldots, \gamma_1^{-1}$ to obtain the net $N$.

   Note, however, that not all sets of cut edges are \emph{``valid''}, in the sense of allowing for an unfolding of the box into a flat net. For example, one can easily see that cutting a single edge of a box never allows for an unfolding, and furthermore in~\Cref{sec:valid_cut_edges} we will describe a simple argument showing one needs at least $4$ cut edges to unfold a box. Similarly, another requirement for a set of cut edges to be valid is not contain cycles; cutting along a cycle would separate the box into disconnected pieces!

Nonetheless, provided an efficient method to find valid sets of cut edges for a box $B$, we can search for common unfoldings of multiple boxes. Indeed, to find a net $N$ that folds into different boxes $B_1, \ldots, B_m$, we need to find a valid set $C_i$ of cut edges in each box $B_i$ such that the different sets $C_i$ are \emph{``compatible''}, which intuitively means that for each pair $B_i, B_j$, the unit squares of $B_i$ can be mapped to those of $B_j$ so that two adjacent squares in $B_i$ without their common edge cut map to two adjacent squares in $B_j$ without their common edge cut. The details of this mapping presented in~\Cref{sec:common_unfoldings}.

Now, to illustrate our general methodology, let us present a concrete result and a high-level sketch of how we obtain it.
\begin{theorem}
  No set of three non-isomorphic boxes of area $58$ or less has a common unfolding.  In other words, $\minPossible(3) > 58$. 
\end{theorem}
\begin{proof}[Methodology]
Let us consider the particular case of ruling out area $46$, since the other areas are analogous. The set of non-isomorphic box-dimensions for area $46$ is 
\(
P(46) = \{ (1, 1, 11), (1, 2, 7), (1, 3, 5)\}.
\)
Any common unfolding of the corresponding boxes is, in particular, a common unfolding of boxes $B_1 := \boxDims{1}{1}{11}$ and $B_2 := \boxDims{1}{2}{7}$, so we focus first on constructing the set $S$ of all common unfoldings of these two boxes, where each common unfolding can be represented as a pair $(C_1, C_2)$, where $C_1$ (resp. $C_2$) is the set of edges to cut in $B_1$ (resp. $B_2$). However, instead of computing $S$ exactly, we compute a superset $S' \supseteq S$, that contains all common unfoldings of $B_1$ and $B_2$, but also potentially pairs $(C_1, C_2)$ where the sets of cut edges do not necessarily allow to unfold the boxes into a common net.
The set $S'$ is obtained by enumerating all satisfying assignments of a CNF formula $\Phi(B_1, B_2)$, whose constraints will be detailed throughout the paper.
Then, for each pair $(C_1, C_2) \in S'$, we do another SAT call to check whether it is possible to unfold box $B_3 := \boxDims{1}{3}{5}$ in a way that is compatible with $(C_1, C_2)$.
Since none of the $|S'|$ calls is satisfiable, we can conclude that no common unfolding of $B_1, B_2, B_3$ exists. \qed
\end{proof}

While a SAT encoding for the problem of finding common unfoldings of two (or more) boxes was already presented by Tadaki and Amano~\cite{tadaki2020searchdevelopmentsboxhaving}, our approach represents a significant improvement in allowing to search and enumerate common unfoldings for significantly larger dimensions. At a high-level, the key improvements of our encoding are:
\begin{enumerate}
  \item When encoding the unfolding of a single box, that is, whether a set of cut edges is \emph{``valid''}, we do not explicitly encode the net as~\cite{tadaki2020searchdevelopmentsboxhaving}, but rather properties that the set of cut edges must satisfy, and moreover, our encoding of these properties is a very efficient under-approximation, that replaces global constraints (i.e., the connectivity constraint of~\cite{tadaki2020searchdevelopmentsboxhaving}) with local constraints.
  \item When encoding a 2-box common unfolding, we do not use a net as intermediary, and instead encode the existence of a direct mapping between the unit-squares of box $B_1$ and those of $B_2$, that \emph{``preserves''} the cut edges. 
  \item We exploit the symmetry of the boxes to reduce the search space. For example, in a $\boxDims{1}{1}{11}$ there are several rotational symmetries that we break.
\end{enumerate}
Before we detail the improvements, let us describe what it means for a set of cut edges to be \emph{``valid''}, or more precisely, some necessary conditions for it.

\section{Valid sets of cut edges}\label{sec:valid_cut_edges}
Let us first briefly describe $4$ well-known (see~\cite[Ch. 21]{demaine2007geometric}, \cite{randomAlg2008}) necessary properties for a set of cut edges to be valid for a box $B$:
\begin{enumerate}
    \item[P1.] \textbf{Connectivity:} The graph induced by the cut edges (taking the set of their endpoints as vertices) must be connected. 
    \item[P2.] \textbf{Cut corners:} The graph induced by the cut edges must touch all the 8 corners of the box $B$.
    \item[P3.] \textbf{Acyclicity:} The graph induced by the cut edges must be acyclic.
    \item[P4.] \textbf{Necessity:} For every set of four unit-squares $\{\sqa, \sqb, \sqc, \sqd\}$ that forms a $2\!\times\!2$ square on $B$, it cannot be the case that \emph{exactly one} edge between these unit-squares is cut, as such cuts are not necessary \cite[Lemma~1]{randomAlg2008}.
\end{enumerate}

We suggest the reader inspects~\Cref{fig:113_box} to check that these constraints are satisfied, and to try to obtain some insight into their necessity. Intuitively, P1 is justified by the fact that the cut edges unfold into the boundary of the net, as exemplified in~\Cref{fig:boxnet}, and that boundary is connected. P2, on the other hand, can be justified by noting that every non-cut edge connecting adjacent unit-squares $\sqa, \sqb$ in a box $B$ will remain an edge connecting two adjacent squares in the net $N$ that folds into $B$. Thus, if none of the three edges incident to a corner are cut, then the three squares incident to that corner will remain adjacent in the net $N$, which is a contradiction since in a polyomino there cannot be three pairwise adjacent squares, as illustrated in \Cref{fig:cut_corners}. 
P3 is justified by the fact that if the graph induced by the cut edges contained a cycle $C$, then the unit squares inside $C$ would be disconnected from the rest when unfolding the box $B$. Finally, P4 is intuitively justified by noticing that if exactly one such edge is cut, then this edge can always be ``glued'' back without affecting the underlying unfolding, a rigorous proof of P4 can also be found in earlier works \cite[Lemma~1]{randomAlg2008}. 

\begin{figure}[b]
    \begin{subfigure}{0.49\linewidth}
    \centering
   \switch{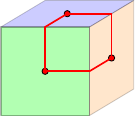}
  \end{subfigure}
  \begin{subfigure}{0.49\linewidth}
    \centering
    \switch{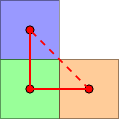}
  \end{subfigure}
  \caption{Illustration of the necessity of the cut-corners property.}\label{fig:cut_corners}
\end{figure}

To represent whether an edge $\{\sqa, \sqb\}$ is cut or not, we can simply use a boolean variable $e_{\sqa, \sqb}$ that is true if and only if $\{\sqa, \sqb\}$ is \emph{not cut}. Then, P2 can be trivially encoded by 8 clauses of the form
\(
    (\overline e_{\sqa, \sqb} \lor \overline e_{\sqa, \sqc} \lor \overline e_{\sqb, \sqc}),
\)
where $\sqa, \sqb, \sqc$ are the three adjacent squares on a corner of the box. Similarly, for each $2\!\times\!2$ square $\{\sqa, \sqb, \sqc, \sqd\}$ on $B$ (with $\{\sqa, \sqd\}$ non-adjacent), $P_4$ can be encoded by 4 clauses of the form
\(
  (e_{\sqa, \sqb} \land e_{\sqa, \sqc} \land e_{\sqb, \sqd} \rightarrow e_{\sqc, \sqd})
\) 
The difficulty, however, arises when encoding properties P1 and P3, which are \emph{``non-local''} properties, and despite a body of research, remain challenging to encode without resulting in either a large number of clauses or poor propagation properties~\cite{gebserSATModuloGraphs2014,zhou_et_al:LIPIcs.SAT.2023.30}.

Our approach replaces these constraints by a set of local constraints that intuitively pursue a similar goal as P1 and P3: ensuring that the graph of cut edges has sufficiently many edges (P1) without having too many (P3). Concretely:
\begin{enumerate}
    \item[1.] To force cutting a significant number of edges, we leverage the work of Tadaki and Amano~\cite{tadaki2020searchdevelopmentsboxhaving}, and assign orientations to each square of the box, which then allows for enforcing consistency constraints between adjacent squares whose common edge is not cut. To satisfy those constraints, a significant number of edges must be cut.
    \item[2.] In contrast to the encoding in~\cite{tadaki2020searchdevelopmentsboxhaving}, which forbid cutting too many edges by enforcing the connectivity of the resulting 2D net, explicit in their encoding, we use a fully novel idea: assigning orientations to the edges of the box, and then forbidding a small number of local directed patterns that every valid unfolding can avoid, but most disconnected nets contain. Intuitively, since disconnected nets correspond to cycles of cut edges, our encoding attempts to prevent such cycles.
\end{enumerate}

%% file: sections/local_constraints.tex

  
  
  

As described in~\Cref{sec:valid_cut_edges}, our goal is to impose constraints that ensure that sufficiently many, but not too many, variables $e_{\sqa, \sqb}$ (representing that the edge $\{\sqa, \sqb\}$ is not cut) are set to true.
We achieve these goals independently: we enforce cutting edges using \emph{``square-orientation constraints''} in a similar (albeit with important differences) way to Tadaki and Amano~\cite{tadaki2020searchdevelopmentsboxhaving}, and use a fully novel approach, based on \emph{``edge directions''} to forbid too many edges from being cut.

\subsection{Square orientations}



\begin{figure}[b]
  \centering
  \begin{subfigure}{0.30\textwidth}
      \centering
      \switch{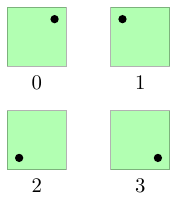}      
      \caption{Square orientations}\label{fig:square_ori}
  \end{subfigure}
  \begin{subfigure}{0.30\textwidth}
          
    \centering 
      \switch{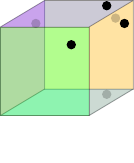}      
    \caption{Box labels.}
    \label{fig:ori_unfold_box}
\end{subfigure}
  \begin{subfigure}{0.30\textwidth}
      \centering
      \switch{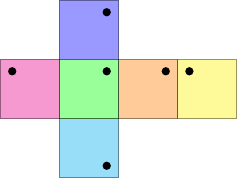}      
      \caption{Induced orientations on net, $r_{\text{green}}(\text{blue}) = 3$.}
      \label{fig:ori_unfold_net}
  \end{subfigure}
  \caption{Square orientations, labels, and relative orientations.}
  \label{fig:orientations}
\end{figure}
Inspired by~\cite{tadaki2020searchdevelopmentsboxhaving}, we think of each square in a box $B$ as having an \emph{``orientation''}, that intuitively represents whether, on an unfolding of the net, the square would be rotated by $0^\circ$, $90^\circ$, $180^\circ$, or $270^\circ$. As Figure \ref{fig:square_ori} indicates, each square is labeled with a dot, which then induces an orientation value on the square once it is unfolded onto the 2D plane. Naturally, this requires a labeling of squares on the box with such dots. Any consistent labeling will work. We adopt the convention that the dot is labeled at the corner that is diagonally the furthest away from the origin when the box is placed in the positive octant. Equivalently, one can directly extend the labeling shown in Figure \ref{fig:ori_unfold_box} to general boxes. Once such a labeling has been fixed, it is clear that any unfolding of the box will induce an orientation assignment $o : B \to  \{1, 2, 3, 4\}$. 

Importantly, any orientation assignment $o \colon B \to \{1, 2, 3, 4\}$ induced by a valid unfolding will necessarily preserve \emph{``relative orientations''} between connected squares. For a pair of connected squares $\sqa, \sqb$, define the relative orientation of $\sqb$ with respect to $\sqa$, $r_{\sqa}(\sqb)$, to be the orientation of $\sqb$ if $\sqa$ was rotated to have orientation $0$ (\Cref{fig:ori_unfold_net}). Note that $r_{\sqa}(\sqb)$ \emph{only} depends on the canonical labeling chosen for the box, and in particular does not depend on potential unfoldings. Since unfolding is an orientation-preserving geometric transformation, for any edge $e = \{\sqa, \sqb\}$ that is \emph{not cut}, the relative orientations between $\sqa, \sqb$ must remain invariant. Thereby necessarily implying $o(\sqb) = o(\sqa) + r_{\sqa}(\sqb)$ (and vice versa), where addition is carried out in $\mathbb{Z}_4$. In fact, these are the only constraints that we enforce in our encoding for square orientations. The constraints are:

\begin{itemize}
  \item Variables $o_{s, d}$ for $s \in B$, $d \in \{1, 2, 3, 4\}$ encoding a function $o : B \to \{1, 2, 3, 4\}$ representing the orientation values. For this to be a well-defined function, we have the following constraints for all $s \in B$. 
      \[\sum_{d = 1}^4 o_{s, d} = 1\]
  \item For each edge $e = \{\sqa, \sqb\}$ in $B$ that is not cut (i.e. $e_{\sqa, \sqb}$ is true), it must be the case that $o(\sqb) = o(\sqa) + r_{\sqa}(\sqb)$ (and vice versa). This is encoded as follows for every $d \in \{1, 2, 3, 4\}$. 
      \[ \left(e_{\sqa, \sqb} \land o_{{s_1}, d} \rightarrow o_{\sqb, r_{\sqa}(\sqb) + d}\right) \land \left(e_{\sqa, \sqb} \land o_{\sqb, d} \rightarrow o_{\sqa, r_{\sqb}(\sqa) + d}\right)\]

  \end{itemize}

\subsection{Edge directions}

Let $G_B$ be the graph with the squares of box $B$ as vertices, and where neighboring squares have a graph edge if and only if their common geometrical edge is not cut. Recall now that forbidding cycles of cut edges is equivalent to making the graph $G_B$ connected, which is our goal. The SAT encoding in~\cite{tadaki2020searchdevelopmentsboxhaving} encodes graph connectivity by choosing one vertex $s^\star \in V(G_B)$ as the source of a \emph{Breadth First Search} (BFS), and then encoding that every vertex is reached by that BFS. 

Concretely, variables $t_{v, k}$ represent that vertex $v$ is reached on step $k$ or earlier of the BFS, with $k$ ranging up to $|V(G_B)| - 1$ in the worst case. Then, the encoding consists of:
\begin{itemize}
  \item The source vertex $s^\star$ is reached at step $0$, enforced by unit clause $t_{s^\star, 0}$.
  \item Each vertex is reached at some step, enforced by the formula $$\bigwedge_{v \in V(G_B)} \bigvee_{k = 0}^{|V(G_B)| - 1} t_{v, k}.$$
  \item Let $N(v)$ denote the set of four neighbors of $v \in V(G_B)$. A vertex $v$ is reached at step at most $k$ if and only if one of its neighbors (or itself) was reached at step at most $k-1$: 
  $$\bigwedge_{k = 1}^{|V(G_B)| - 1}\bigwedge_{v \in V(G_B)} \left( t_{v, k} \leftrightarrow \bigvee_{u \in N(v)} (t_{u, k-1} \land e_{u,v})\right).$$
\end{itemize}

As a result, the number of variables and clauses in their encoding is quadratic in $|V(G_B)|$, the number of squares of the box. It is possible to upper-bound $k \leq T$ for some $T \leq |V(G_B)| - 1$ to improve performance at the cost of potentially missing solutions, thus such bounds cannot be used if one wants to enumerate all solutions. \Cref{tab:BFS_times} illustrates this for a typical sub-problem (\Cref{sec:symmetry}) between boxes $\boxDims{1}{1}{7}, \boxDims{1}{3}{3}$.  

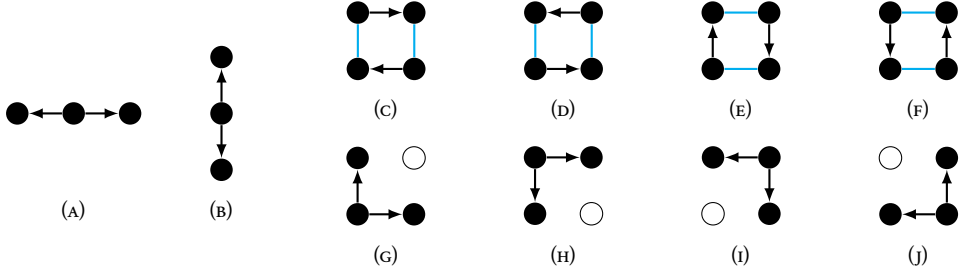
\begin{figure}[t]
  \centering
  \begin{minipage}{.3\textwidth}
  \begin{subfigure}{0.6\textwidth}
      \centering
      \begin{tikzpicture}[scale=0.75]
          \node[circle, fill=white, inner sep=3pt] (X) at (0,-1) {};
          \node[circle, fill=black, inner sep=3pt] (A) at (0,0) {};
          \node[circle, fill=black, inner sep=3pt] (B) at (1,0) {};
          \node[circle, fill=black, inner sep=3pt] (C) at (2,0) {};
          \node[circle, fill=white, inner sep=3pt] (Y) at (0,1) {};
          \draw[-latex,thick] (B) -- (A);
          \draw[-latex,thick] (B) -- (C); 
      \end{tikzpicture}
      \caption{}\label{fig:forbidden-subgraphs-a-body}
  \end{subfigure}
  \begin{subfigure}{0.3\textwidth}
      \centering
      \begin{tikzpicture}[scale=0.75]
          \node[circle, fill=black, inner sep=3pt] (A) at (0,-1) {};
          \node[circle, fill=black, inner sep=3pt] (B) at (0,0) {};
          \node[circle, fill=black, inner sep=3pt] (C) at (0,1) {};
          \draw[-latex,thick] (B) -- (A);
          \draw[-latex,thick] (B) -- (C); 
      \end{tikzpicture}
      \caption{}\label{fig:forbidden-subgraphs-b-body}
  \end{subfigure}
  \end{minipage}
  \begin{minipage}{.68\textwidth}
  \begin{subfigure}{0.24\textwidth}
      \centering
      \begin{tikzpicture}[scale=0.75]
          \node[circle, fill=black, inner sep=3pt] (A) at (0,0) {};
          \node[circle, fill=black, inner sep=3pt] (B) at (1,0) {};
          \node[circle, fill=black, inner sep=3pt] (C) at (0,1) {};
          \node[circle, fill=black, inner sep=3pt] (D) at (1,1) {};
          \draw[-latex,thick] (B) -- (A);
          \draw[-,thick, cyan] (B) -- (D);
          \draw[-,thick, cyan] (A) -- (C); 
          \draw[-latex,thick] (C) -- (D); 
      \end{tikzpicture}
      \caption{}\label{fig:forbidden-subgraphs-c-body}
  \end{subfigure}
  \begin{subfigure}{0.24\textwidth}
      \centering
      \begin{tikzpicture}[scale=0.75]
          \node[circle, fill=black, inner sep=3pt] (A) at (0,0) {};
          \node[circle, fill=black, inner sep=3pt] (B) at (1,0) {};
          \node[circle, fill=black, inner sep=3pt] (C) at (0,1) {};
          \node[circle, fill=black, inner sep=3pt] (D) at (1,1) {};
          \draw[-latex,thick] (A) -- (B);
          \draw[-,thick, cyan] (D) -- (B);
          \draw[-,thick, cyan] (C) -- (A); 
          \draw[-latex,thick] (D) -- (C); 
      \end{tikzpicture}
      \caption{}\label{fig:forbidden-subgraphs-d-body}
  \end{subfigure}
      \begin{subfigure}{0.24\textwidth}
          \centering
          \begin{tikzpicture}[scale=0.75]
              \node[circle, fill=black, inner sep=3pt] (A) at (0,0) {};
              \node[circle, fill=black, inner sep=3pt] (B) at (1,0) {};
              \node[circle, fill=black, inner sep=3pt] (C) at (0,1) {};
              \node[circle, fill=black, inner sep=3pt] (D) at (1,1) {};
              \draw[-,thick, cyan] (A) -- (B);
              \draw[-latex,thick] (D) -- (B);
              \draw[-latex,thick] (A) -- (C); 
              \draw[-,thick, cyan] (D) -- (C); 
          \end{tikzpicture}
          \caption{}\label{fig:forbidden-subgraphs-e-body}
      \end{subfigure}
          \begin{subfigure}{0.24\textwidth}
              \centering
              \begin{tikzpicture}[scale=0.75]
                  \node[circle, fill=black, inner sep=3pt] (A) at (0,0) {};
                  \node[circle, fill=black, inner sep=3pt] (B) at (1,0) {};
                  \node[circle, fill=black, inner sep=3pt] (C) at (0,1) {};
                  \node[circle, fill=black, inner sep=3pt] (D) at (1,1) {};
                  \draw[-,thick, cyan] (A) -- (B);
                  \draw[-latex,thick] (B) -- (D);
                  \draw[-latex,thick] (C) -- (A); 
                  \draw[-,thick, cyan] (D) -- (C); 
              \end{tikzpicture}
          \caption{}\label{fig:forbidden-subgraphs-f-body}
       \end{subfigure}
      
       \vspace{0.7em}

       \begin{subfigure}{0.24\textwidth}
        \centering
        \begin{tikzpicture}[scale=0.75]
            \node[circle, fill=black, inner sep=3pt] (A) at (0,0) {};
            \node[circle, fill=black, inner sep=3pt] (B) at (1,0) {};
            \node[circle, fill=black, inner sep=3pt] (C) at (0,1) {};
            \node[circle, draw, inner sep=3pt] (D) at (1,1) {};
            \draw[-latex,thick] (A) -- (B);
            \draw[-latex,thick] (A) -- (C); 
        \end{tikzpicture}
    \caption{}\label{fig:forbidden-subgraphs-g-body}
    \end{subfigure}
    \begin{subfigure}{0.24\textwidth}
      \centering
      \begin{tikzpicture}[scale=0.75]
          \node[circle, fill=black, inner sep=3pt] (A) at (0,0) {};
          \node[circle, draw, inner sep=3pt] (B) at (1,0) {};
          \node[circle, fill=black, inner sep=3pt] (C) at (0,1) {};
          \node[circle, fill=black, inner sep=3pt] (D) at (1,1) {};
          \draw[-latex,thick] (C) -- (A); 
          \draw[-latex,thick] (C) -- (D); 
      \end{tikzpicture}
  \caption{}\label{fig:forbidden-subgraphs-h-body}
  \end{subfigure}
  \begin{subfigure}{0.24\textwidth}
    \centering
    \begin{tikzpicture}[scale=0.75]
        \node[circle, draw, inner sep=3pt] (A) at (0,0) {};
        \node[circle, fill=black, inner sep=3pt] (B) at (1,0) {};
        \node[circle, fill=black, inner sep=3pt] (C) at (0,1) {};
        \node[circle, fill=black, inner sep=3pt] (D) at (1,1) {};
        \draw[-latex,thick] (D) -- (B);
        \draw[-latex,thick] (D) -- (C); 
    \end{tikzpicture}
\caption{}\label{fig:forbidden-subgraphs-i-body}
\end{subfigure}
\begin{subfigure}{0.24\textwidth}
  \centering
  \begin{tikzpicture}[scale=0.75]
      \node[circle, fill=black, inner sep=3pt] (A) at (0,0) {};
      \node[circle, fill=black, inner sep=3pt] (B) at (1,0) {};
      \node[circle, draw, inner sep=3pt] (C) at (0,1) {};
      \node[circle, fill=black, inner sep=3pt] (D) at (1,1) {};
      \draw[-latex,thick] (B) -- (A);
      \draw[-latex,thick] (B) -- (D);
  \end{tikzpicture}
\caption{}\label{fig:forbidden-subgraphs-j-body}
\end{subfigure}
\end{minipage}
  \caption{The set $\mathcal{F}$ of forbidden subgraphs. The blue edges represent that the graph belongs to $\mathcal{F}$ regardless of the orientation of its blue edges.
  For the forbidden patterns \ref{fig:forbidden-subgraphs-g-body} to \ref{fig:forbidden-subgraphs-j-body}, the white node must not be present in the graph.
  Blocking the patterns \ref{fig:forbidden-subgraphs-a-body} to \ref{fig:forbidden-subgraphs-b-body} requires one binary clause per pattern. Blocking the patterns \ref{fig:forbidden-subgraphs-c-body} to \ref{fig:forbidden-subgraphs-f-body} requires two ternary clauses. Finally, blocking the patterns \ref{fig:forbidden-subgraphs-g-body} to \ref{fig:forbidden-subgraphs-j-body} requires two ternary clauses. The ternary clauses make use of the observation that it is impossible to preserve exactly 3 of the 4 edges. 
  }
  \label{fig:forbidden-subgraphs-body}
\end{figure}

We approximate connectivity using a constant number of clauses per square. First, for each neighboring pair of squares $(\sqa, \sqb)$, we create two variables, $d_{\sqa, \sqb}$ and $d_{\sqb, \sqa}$ representing that the preserved edge $\{\sqa, \sqb\}$ will be directed from $\sqa$ toward $\sqb$ (or from $\sqb$ toward $\sqa$, respectively). We have clauses $(\overline d_{\sqa, \sqb} \lor \overline d_{\sqb, \sqa})$ to prevent both directions per preserved edge. 
If an edge $\{\sqa, \sqb\}$ is preserved, it has at least one direction, this is enforced by $(e_{\sqa, \sqb} \rightarrow d_{\sqa, \sqb} \lor d_{\sqb, \sqa})$. Finally, if the edge $\{\sqa, \sqb\}$ has a direction, it is necessarily preserved. This is enforced by $d_{\sqa, \sqb} \rightarrow e_{\sqa, \sqb}$ and $d_{\sqb, \sqa} \rightarrow e_{\sqa, \sqb}$. 

Now, we encode that a special vertex $s^\star$ is a \emph{unique sink} by enforcing that (i) it has outdegree $0$ according to the edge directions, so $\bigwedge_{u \in N(s^\star)} \overline d_{s^\star, u}$,
and (ii) every other vertex has outdegree at least $1$:
$$\bigwedge_{v \in V(G_B) \setminus \{s^\star\}} \left.\bigvee_{u \in N(v)} d_{v, u} \right.$$
Note that for us, $s^\star$ is a \emph{``sink''}, instead of a \emph{``source''}. To approximate further that our edge directions correspond to a \emph{reverse BFS} from $s^\star$, we forbid all local patterns depicted in~\Cref{fig:forbidden-subgraphs-body} as subgraphs. Since this forbids a constant number of possibilities around each vertex, it totals $O(|V(G_B)|)$ clauses.
While this is insufficient in theory to guarantee connectivity, it almost always results in connected nets in practice. On the other hand, we prove in~\Cref{thm:forbidden-patterns} that these constraints are sound, meaning that every valid unfolding must satisfy them. 

\begin{table}[h]
  \centering
  \caption{Runtimes with BFS/forbidding local patterns at $|V(G_B)| = 30$.}
  \vspace{0.1em}
  \begin{tabular}{cccccc|c}
  \toprule
  ~~$T = 15$~~ & ~~$T = 16$~~ & ~~$T = 17$~~ & ~~$T = 18$~~ & ~~$T = 19$~~ & ~~$T = 20$~~ & ~~\bf{Local patterns}~~ \\
   \midrule
    79.66\,s & 240.53\,s & 385.56\,s & 504.43\,s & 947.83\,s & 1823.88\,s & \bf{36.51\,s}\\
  \bottomrule
  \end{tabular}
  \label{tab:BFS_times}
\end{table}

%% file: sections/common_unfoldings.tex
An important feature of our approach is the \emph{implicit} representation of nets via their corresponding cut edges, in contrast to \emph{explicitly} representing them as subsets of the 2D plane. Utilizing such implicit representations, the search for a net $N$ that folds into different boxes $B_1, \cdots, B_m$ naturally reduces down to the search for cut edges $C_1, \cdots, C_m$ that are ``compatible''. Importantly, such constraints are \emph{intrinsic} to the boxes and cut edges $(B_i, C_i)$ and do not concern the explicit 2D representation of nets. This allows for much more compact and efficient encodings compared to earlier work \cite{tadaki2020searchdevelopmentsboxhaving}, which uses explicit representations.

\begin{figure}[b]   
    \centering
    \switch{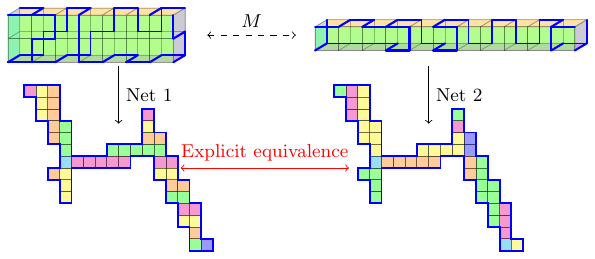}

    \caption{Equivalence of nets of boxes $\boxDims{1}{2}{7}$ and $\boxDims{1}{1}{11}$, explicitly via nets (red) and implicitly via cut edges ($M$).}
    \label{fig:implicit_equivalence}
\end{figure}

Figure \ref{fig:implicit_equivalence} shows how the search for a common net can be done explicitly by enforcing equivalence on the resulting 2D unfoldings, which is costly and introduces additional auxiliary variables. In contrast, we achieve this by encoding an equivalence mapping $M \colon B_2 \to B_1$ directly on the level of boxes, and thereby alleviating the need to consider the nets explicitly. 

\begin{figure}[t]
    \centering
    \switch{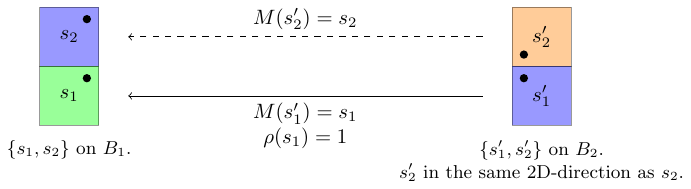}

%
%
%
%
%
%
    \caption{Constraints on $M$ with relative change in orientation $\rho$.}
    \label{fig:equivalence_constraints}
\end{figure}
\Cref{fig:equivalence_constraints} depicts the constraints imposed on $M \colon B_2 \to B_1$. If a square $\boxb{\sqa} \in B_2$ is mapped to $\sqa \in B_1$ \emph{and} the edge $e = \{\sqa, \sqb\}$ is not cut, then it must necessarily be the case that the corresponding neighbor of $\boxb{\sqa}$, in this case $\boxb{\sqb}$, maps to $\sqb$ and the edge $e = \{\boxb{\sqa}, \boxb{\sqb}\}$ is \emph{not cut} in box 2. Furthermore, notice that $\boxb{\sqb}$ is the unique neighbor of $\boxb{\sqa}$ in the same direction as $\sqb$, and can be computed from $o(\boxb{\sqa}) - o(\sqa)$, the relative change in orientations. These constraints are encoded as follows. 

\begin{itemize}
    \item Mapping variables $m_{s, \boxb{s}}$ for $\boxb{s} \in B_2, s \in B_1$ that encode a bijection $M \colon B_2 \to B_1$. For each $\boxb{s} \in B_2$ we have the constraint
        \[\sum_{s \in B_1} m_{s, \boxb{s}} = 1\]
    \item Auxiliary variables $\rho_{s, d}$ for $s \in B_1, d \in \{1,2,3,4\}$ that encode a function $\rho \colon B_1 \to [4]$ indicating the change in orientation between $s$ and the unique $\boxb{s} \in B_2$ where $M(\boxb{s}) = s$. I.e. $o(\boxb{s}) = o(s) + \rho(s)$ holds. This is encoded via the following where $\boxb{s} \in B_2, s \in B_1$ and $d, r \in \{1,2,3,4\}$.
        \[m_{s, \boxb{s}} \land o_{\boxb{s}, d + r} \land o_{s, d} \rightarrow \rho_{s, r}\]
    \item Constraints that enforce the preservation of relative positions as shown in Figure \ref{fig:equivalence_constraints}. Let $\{\sqa, \sqb\}$ be an edge in $B_1$, $\boxb{\sqa} \in B_2$ and $d, r \in \{1,2,3,4\}$ be such that $\rho(\sqa) = r$ and $d$ indicates the 2D-direction of $\sqb$. Let $\boxb{\sqb}$ be the unique neighbor of $\boxb{\sqa}$ in the same direction as $\sqb$. The constraints are:
        \begin{align*}
            &m_{\sqa, \boxb{\sqa}} \land \rho_{\sqa, r} \land e_{\sqa, \sqb} \rightarrow m_{\sqb, \boxb{\sqb}}\\
            &m_{\sqa, \boxb{\sqa}} \land \rho_{\sqa, r} \land e_{\sqa, \sqb} \rightarrow e_{\boxb{\sqa}, \boxb{\sqb}}\\
            &m_{\sqa, \boxb{\sqa}} \land \rho_{\sqa, r} \land e_{\boxb{\sqa}, \boxb{\sqb}} \rightarrow e_{\sqa, \sqb}
        \end{align*}
        where the latter two constraints encode that the edge $e = \{\sqa, \sqb\}$ is cut if and only if $e' = \{\sqa', \sqb'\}$ is cut. Only the first type of constraint is necessary. The other two types help with performance.  
    \item Finally, we write $\text{EQUIV}_{B_1, B_2}$ to denote the overall encoding for equivalence obtained by taking the conjunction of all constraints above. Equivalence of more than 2 boxes can easily be encoded by enforcing $\text{EQUIV}_{B_1, B_n}$ for all $n \geq 2$.
\end{itemize}

%% file: sections/symmetry_breaking.tex
Geometrically, if a box $B$ unfolds into a net $N$, then the symmetries of $B$ can certainly be unfolded into the same net $N$. To efficiently enumerate all unfoldings, it is important to avoid the computation of repeated solutions by breaking these symmetries. In this section, we describe how symmetry breaking is carried out and show experimentally that this greatly improves the performance.

As a first step, notice that a geometric (2D) rotation of a net only affects the orientations of squares and preserves the edges. Indeed, constraints on the orientation of squares introduced in Section \ref{sec:local} only enforce relative equality, and therefore if $o \colon B \to \{1,2,3,4\}$ gives a satisfying assignment of orientations, then a constant shift (e.g. $o'(s) = o(s) + 1$ for all $s$) will also be satisfying and corresponds to a 2D rotation of the underlying net. Thus, we may without loss of generality pick a distinguished square $\hat{s} \in B$ and enforce $o(\hat{s}) = 0$ with a unit clause. 

\begin{figure}[b]
    \centering
    \switch{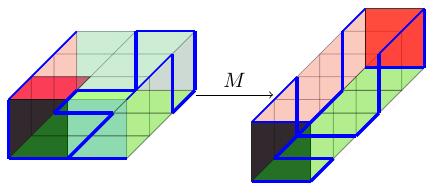}
    \caption{Image of pairs of squares under $M$, mapping the (black, red) pair on $B_2$ to the (black, red) pair on $B_1$.}
    \label{fig:mapping_pair}
\end{figure}

To further break symmetries when searching for common unfoldings (of two boxes $B_1, B_2$), consider the image of pairs of squares under the equivalence mapping $M$ (Section \ref{sec:common_unfoldings}) of a common unfolding. First, fix a pair of squares $(\sqa, \sqb)$ on $B_1$. Since $M$ is a bijection, some pair of squares $(\boxb{\sqa}, \boxb{\sqb})$ on $B_2$ satisfies $(M(\boxb{\sqa}), M(\boxb{\sqb})) = (\sqa, \sqb)$, as illustrated in Figure \ref{fig:mapping_pair}. If $Q \in \text{Sym}(B_2)$ is a symmetry of $B_2$, then certainly the same unfolding can be carried out on $Q(B_2)$ (as $Q$ is a symmetry), inducing a common unfolding between $Q(B_2)$ and $B_1$ where the pair $(Q^{-1}(\boxb{\sqa}), Q^{-1}(\boxb{\sqb}))$ now maps to $(\sqa, \sqb)$. Therefore, if $P_2$ denotes the set of all pairs of $B_2$ up to symmetry, then the search for common unfoldings can be reduced to the cases $(M(\boxb{\sqa}), M(\boxb{\sqb})) = (\sqa, \sqb)$ for every pair $(\boxb{\sqa}, \boxb{\sqb}) \in P_2$. Furthermore, if $(\sqa, \sqb)$ are chosen on $B_1$ such that $\sqa$ is in the orbit of $\sqb$ under symmetries of $B_1$, then it suffices to consider the pairs in $P_2$ as unordered, thereby reducing the number of cases by a factor of 2. Table \ref{fig:pairs_table} depicts the number of such pairs relative to the total number of pairs ${\text{Area}}\choose {2}$. As shown, the number of pairs up to symmetry is relatively small compared to the number of all pairs. 

\begin{table}[t]
    \centering
    \caption{Number of pairs}
    \begin{tabular}{@{}c l@{\hspace{10pt}}l cc@{}}
    \toprule
      ~Area~ & \multicolumn{2}{c}{~Dimensions~} & $\#$Pairs~~~ & Total\%~~ \\
     \midrule
     22 & $1 \! \times \! 1 \! \times \! 5$ & $1 \! \times \! 2 \! \times \! 3$ & 45 &  9.74\%\\
    
     30 & $1 \! \times \! 1 \! \times \! 7$ & $1 \! \times \! 3 \! \times \! 3$ & 47 & 5.40\% \\
    
     34 & $1 \! \times \! 1 \! \times \! 8$ & $1 \! \times \! 2 \! \times \! 5$ & 97 & 8.64\% \\
    
     38 & $1 \! \times \! 1 \! \times \! 9$ & $1 \! \times \! 3 \! \times \! 4$ & 120 & 8.53\% \\
    
     42 & $1 \! \times \! 1 \! \times \! 10$ & $2 \! \times \! 3 \! \times \! 3$ & 77 & 4.47\% \\
    
     46 & $1 \! \times \! 1 \! \times \! 11$ & $1 \! \times \! 2 \! \times \! 7$ & 169 & 8.16\% \\
    
     54 & $1 \! \times \! 1 \! \times \! 13$ & $1 \! \times \! 3 \! \times \! 6$ & 231 & 8.07\% \\
    
     58 & $1 \! \times \! 1 \! \times \! 14$ & $1 \! \times \! 2 \! \times \! 9$ & 261 & 7.89\% \\
    \bottomrule
    \end{tabular}
    \label{fig:pairs_table}
\end{table}

Finally, when the dimensions of $B_1$ is of the form $\boxDims{1}{1}{n}$, we always pick the pair $(\sqa, \sqb)$ to be the $1\!\times\!1$ faces. In addition to them being symmetric, $B_1$ has full rotational symmetry about these faces and therefore the orientation of $\sqa$ can be altered by applying such rotations \emph{without rotating the underlying net}. Consequently, this implies that we may always assume without loss of generality that $o(\sqa) = 0$, in addition to fixing the orientation of a square $\boxb{\sqa} \in B_2$. To summarize, when enumerating common unfoldings between two boxes $B_1, B_2$, the following symmetry-breaking procedure is carried out.

\begin{enumerate}
    \item Fix a pair of squares $(\sqa, \sqb)$ on $B_1$ that belong to the same equivalence class under symmetry. 
    \item If $B_1$ is of the form $\boxDims{1}{1}{n}$, choose $(\sqa, \sqb)$ to be the $1\!\times\!1$ faces and enforce $o(\sqa) = 0$ using a unit clause. 
    \item Let $P_2$ be the set of (unordered) pairs of squares on $B_2$ unique up to symmetry. For each pair $(\boxb{\sqa}, \boxb{\sqb}) \in P_2$, encode $M(\boxb{\sqa}) = \sqa, M(\boxb{\sqb}) = \sqb, o(\boxb{\sqa}) = 0$ using 3 unit clauses.
    \item Solve the corresponding sub-problem for each pair $(\boxb{\sqa}, \boxb{\sqb}) \in P_2$. Note that since each sub-problem is independent, they can be run in parallel. 
\end{enumerate}

%% file: sections/experimental.tex
    \begin{figure}[b]
        \centering
            \switch{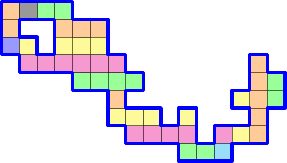}
        \caption{Overlapping net (left) and touching net (right) of $\boxDims{1}{1}{11}$.}
        \label{fig:touch_overlap_nets}
    \end{figure}

Using the encodings given in \Cref{sec:local,sec:common_unfoldings} without symmetry-breaking, we were able to compute unfoldings between all pairs of boxes with equal area up to $86$ (\Cref{tab:all_2_unfold}), thereby establishing $\minImpossible(2) > 86$. As indicated in \Cref{tab:all_2_unfold}, most of such unfoldings with high area were previously unknown. 
\input{sections/tableres}
\begin{table}[t!]
    \caption{Exhaustive enumeration of solutions. The columns labeled $|S|$, $|T|$, and $|O|$ show the 
    number of simple, touching, and overlapping solutions, respectively.}
        \label{tab:all_sols_times}
    \vspace{5pt}
    \begin{tabular}{@{}c l@{\hspace{6pt}}l rrrrrrc@{}}
     \toprule
      Area~ & \multicolumn{2}{c}{~Dimensions~} & ~~$\#$SAT & Unique & $|S|$~~ & $|T|$ & $|O|$ & ~Time~~~~~~~ & First\\
     \midrule
     22 & $1 \! \times \! 1 \! \times \! 5$ & $1 \! \times \! 2 \! \times \! 3$ & 3\,942 & 2\,303 & 2\,263 & 27 & 13 & 2 mins (local) & \cite{22allsols2011}\\
     30 & $1 \! \times \! 1 \! \times \! 7$ & $1 \! \times \! 3 \! \times \! 3$ & 1\,790 & 1\,080 & 1\,070 & 10 & 0 & \,~10 mins (local) & \cite{XU20171}\\
     34 & $1 \! \times \! 1 \! \times \! 8$ & $1 \! \times \! 2 \! \times \! 5$ & 131\,054 & 35\,700 & 35\,675 & 22 & 3 & 1 hour (local) & $\checkmark$\\
     38 & $1 \! \times \! 1 \! \times \! 9$ & $1 \! \times 3 \! \times \! 4$ & 5\,854 & 4\,509 & 4\,469 & 36 & 4 & 6 hrs (local) & $\checkmark$\\
     42 & $1 \! \times \! 1 \! \times \! 10$ & $2 \! \times 3 \! \times \! 3$ & 128\,558 & 111\,948 & 111\,387 & \,~559 & 2 & 1 day (local) & $\checkmark$\\
     46 & $1 \! \times \! 1 \! \times \! 11$ & $1 \! \times \! 2 \! \times \! 7$ & 16\,928 & 15\,236 & 14\,971 & 16 & \,~249 & 3 hrs (PSC) & $\checkmark$\\
     54 & $1 \! \times \! 1 \! \times \! 13$ & $1 \! \times \! 3 \! \times \! 6$ & 56\,087 & 51\,884 & 51\,836 & 48 & 0 & 1 day (PSC) & $\checkmark$\\
     58 & $1 \! \times \! 1 \! \times \! 14$ & $1 \! \times \! 2 \! \times \! 9$ & \,~2\,150\,373 & \,~551\,935 & \,~551\,923 & 7 & 5 & 2 days (PSC) & $\checkmark$\\
     \bottomrule
    \end{tabular}
    \end{table}
    
    With the symmetry-breaking procedure described in~\Cref{sec:symmetry}, we were able to completely enumerate all solutions for the pairs of boxes shown in~\Cref{tab:all_sols_times}. Local computations were conducted using a standard laptop, computations labeled with ``cluster'' were conducted on a supercomputer~\cite{brownBridges2PlatformRapidlyEvolving2021} running 64 SAT solvers in parallel. Solutions were obtained by enumerating all solutions for each fixed pair of squares on the second box (\Cref{sec:symmetry}), using an allsat variant of CaDiCaL \cite{BiereFallerFazekasFleuryFroleyks-CAV24}, available at \url{https://github.com/jreeves3/allsat-cadical}. Through these computations, we found unfoldings with diameter close to the surface area (see~\Cref{sec:extra_unfolding}) making it difficult to compute using a BFS-style encoding for connectivity. 
    
    In addition to standard nets (simply connected, no overlaps), our encoding is also capable of finding non-standard unfoldings, which we classify into ``touching'' and ``overlapping''. Both are nets in the sense that one can obtain them by unfolding a box, however they are non-standard as the resulting shape is not simply-connected. The gray square in the overlapping net shown in~\Cref{fig:touch_overlap_nets} is an overlapping green and yellow square. Touching nets can always fold into the original box if cuts were allowed, e.g. the touching net shown in~\Cref{fig:touch_overlap_nets} can fold into $\boxDims{1}{1}{11}$ if a cut is made along the red edge. 

Beyond common unfoldings, we also enumerated all nets for boxes of the form $\boxDims{1}{1}{n}$ for $n \leq 4$ (\Cref{tab:11n_nets}). To the best of our knowledge, the number of such nets was previously unknown except for $n = 1$. Such computations were carried out locally on a standard laptop. 
\begin{table}[t]
    \centering
    \caption{Exhaustive enumeration of $\boxDims{1}{1}{n}$ nets}
     \label{tab:11n_nets}
    \begin{tabular}{@{}crrrrrr@{}}
     \toprule
       ~Dimensions~ & ~~~~$\#$SAT & ~~Unique~ & ~~~$|S|$ & ~$|T|$ & ~$|O|$ & ~~~Time~ \\
     \midrule
     $ \boxDims{1}{1}{1}$ & 384 & 11 & 11 & 0 & 0 & 0.05\,s\\
      $\boxDims{1}{1}{2}$ & 12\,124 & 723 & 723 & 0 & 0 & 1.79\,s\\
     $\boxDims{1}{1}{3}$ & 240\,304 & 15\,061 & 14\,978 & 79 & 4 & 77.78\,s\\
    $\boxDims{1}{1}{4}$ & 3\,708\,380 & 231\,310 & 228\,547 & 2\,603 & 160 & 9 hrs\\
     \bottomrule
    \end{tabular}
\end{table}

%% file: sections/tableres.tex
The first missing entry is the pair $\boxDims{1}{2}{14}, \boxDims{2}{2}{10}$ at area 88. It is worth noting that common unfoldings can always be scaled up by subdividing squares, e.g. a common unfolding of $\boxDims{1}{1}{5}, \boxDims{1}{2}{3}$ directly implies the existence of a common unfolding for $\boxDims{2}{2}{10}, \boxDims{2}{4}{6}$ by dividing each square into 4 sub-squares \cite{randomAlg2008}. 
\begin{table}[t!]
    \centering
\caption{Existence of common unfoldings.}\label{tab:all_2_unfold}
\vspace{5pt}
\begin{tabular}{@{}c l@{\hspace{5pt}}lc@{}}
    \toprule
    Area\! & \multicolumn{2}{c}{~Dimensions~} & \!First\\
    \midrule
    22 & $\boxDims{1}{1}{5}$ & $\boxDims{1}{2}{3}$ &  \cite{randomAlg2008}\\
    30 & $\boxDims{1}{1}{7}$ & $\boxDims{1}{3}{3}$ & \cite{randomAlg2008}\\
    34 &$\boxDims{1}{1}{8}$ & $\boxDims{1}{2}{5}$  & \cite{randomAlg2008}\\
    38 & $\boxDims{1}{1}{9}$ & $\boxDims{1}{3}{4}$ &  \cite{randomAlg2008}\\
    40 & $\boxDims{1}{2}{6}$ & $\boxDims{2}{2}{4}$ &  \cite{tadaki2020searchdevelopmentsboxhaving}\\
    42 & $\boxDims{1}{1}{10}$ & $\boxDims{2}{3}{3}$ &  $\checkmark$\\
    46 & $\boxDims{1}{1}{11}$ & $\boxDims{1}{3}{5}$ &  \cite{randomAlg2008}\\
    46 & $\boxDims{1}{1}{11}$ & $\boxDims{1}{2}{7}$ &  $\checkmark$\\
    46 & $\boxDims{1}{2}{7}$ & $\boxDims{1}{3}{5}$ & \cite{randomAlg2008}\\
    48 & $\boxDims{1}{4}{4}$ & $\boxDims{2}{2}{5}$ &  $\checkmark$\\
    54 & $\boxDims{1}{1}{13}$ & $\boxDims{3}{3}{3}$ &  \cite{randomAlg2008}\\
    54 & $\boxDims{1}{1}{13}$ & $\boxDims{1}{3}{6}$ &  \cite{randomAlg2008}\\
    54 & $\boxDims{1}{3}{6}$ & $\boxDims{3}{3}{3}$ &  \cite{randomAlg2008}\\

    58 &$\boxDims{1}{1}{14}$ & $\boxDims{1}{2}{9}$ &  $\checkmark$\\
    58 &$\boxDims{1}{1}{14}$ & $\boxDims{1}{4}{5}$ &  \cite{randomAlg2008}\\
    58 &$\boxDims{1}{2}{9}$ & $\boxDims{1}{4}{5}$ &  $\checkmark$\\
    62 &$\boxDims{1}{1}{15}$ & $\boxDims{1}{3}{7}$ &  $\checkmark$\\
    62 &$\boxDims{1}{1}{15}$ & $\boxDims{2}{3}{5}$ &  $\checkmark$\\
    62 & $\boxDims{1}{3}{7}$ & $\boxDims{2}{3}{5}$ &  \cite{randomAlg2008}\\
    64 & $\boxDims{1}{2}{10}$ & $\boxDims{2}{2}{7}$ &  \cite{randomAlg2008}\\
    64 & $\boxDims{1}{2}{10}$ & $\boxDims{2}{4}{4}$ &  $\checkmark$\\
    64 & $\boxDims{2}{2}{7}$ & $\boxDims{2}{4}{4}$ &  \cite{randomAlg2008}\\
    66 & $\boxDims{1}{1}{16}$ & $\boxDims{3}{3}{4}$ &  $\checkmark$\\
    70 & $\boxDims{1}{1}{17}$ & $\boxDims{1}{2}{11}$ &  $\checkmark$\\
    70 & $\boxDims{1}{1}{17}$ & $\boxDims{1}{3}{8}$ &  $\checkmark$\\
    70 & $\boxDims{1}{1}{17}$ & $\boxDims{1}{5}{5}$ &  \cite{randomAlg2008}\\
    70 & $\boxDims{1}{2}{11}$ & $\boxDims{1}{3}{8}$ &  \cite{randomAlg2008}\\
    70 & $\boxDims{1}{2}{11}$ & $\boxDims{1}{5}{5}$ &  $\checkmark$\\
    70 & $\boxDims{1}{3}{8}$ & $\boxDims{1}{5}{5}$ &  $\checkmark$\\
    72 & $\boxDims{2}{2}{8}$ & $\boxDims{2}{3}{6}$ &  $\checkmark$\\
    76 & $\boxDims{1}{2}{12}$ & $\boxDims{2}{4}{5}$ &  $\checkmark$\\
     \bottomrule
     
     \end{tabular}
    \begin{tabular}{@{}c l@{\hspace{6pt}}lc@{}}
    \toprule
    Area\! & \multicolumn{2}{c}{~Dimensions~} & \!First\\
    \midrule
    %
    78 & $\boxDims{1}{1}{19}$ & $\boxDims{1}{3}{9}$ &  $\checkmark$\\
    78 & $\boxDims{1}{1}{19}$ & $\boxDims{1}{4}{7}$ &  $\checkmark$\\
    78 & $\boxDims{1}{1}{19}$ & $\boxDims{3}{3}{5}$ &  $\checkmark$\\
    78 & $\boxDims{1}{3}{9}$ & $\boxDims{1}{4}{7}$ &  $\checkmark$\\
    78 & $\boxDims{1}{3}{9}$ & $\boxDims{3}{3}{5}$ &  $\checkmark$\\
    78 & $\boxDims{1}{4}{7}$ & $\boxDims{3}{3}{5}$ &  $\checkmark$\\
    80 & $\boxDims{2}{2}{9}$ & $\boxDims{3}{4}{4}$ &  $\checkmark$\\
    82 & $\boxDims{1}{1}{20}$ & $\boxDims{1}{2}{13}$ &  $\checkmark$\\
    82 & $\boxDims{1}{1}{20}$ & $\boxDims{1}{5}{6}$ &  $\checkmark$\\
    82 & $\boxDims{1}{1}{20}$ & $\boxDims{2}{3}{7}$ &  $\checkmark$\\
    82 & $\boxDims{1}{2}{13}$ & $\boxDims{1}{5}{6}$ &  $\checkmark$\\
    82 & $\boxDims{1}{2}{13}$ & $\boxDims{2}{3}{7}$ &  $\checkmark$\\
    82 & $\boxDims{1}{5}{6}$ & $\boxDims{2}{3}{7}$ &  $\checkmark$\\
    86 & $\boxDims{1}{1}{21}$ & $\boxDims{1}{3}{10}$ &  $\checkmark$\\
    86 & $\boxDims{1}{2}{14}$ & $\boxDims{1}{4}{8}$ &  $\checkmark$\\
    86 & $\boxDims{1}{2}{14}$ & $\boxDims{2}{4}{6}$ &  $\checkmark$\\
    88 & $\boxDims{1}{4}{8}$ & $\boxDims{2}{2}{10}$ &  \cite{randomAlg2008}\\
    88 & $\boxDims{1}{4}{8}$ & $\boxDims{2}{4}{6}$ &  $\checkmark$\\
    88 & $\boxDims{2}{2}{10}$ & $\boxDims{2}{4}{6}$ &  \cite{randomAlg2008}\\
    90 & $\boxDims{2}{5}{5}$ & $\boxDims{3}{3}{6}$ &  $\checkmark$\\
    94 & $\boxDims{1}{5}{7}$ & $\boxDims{3}{4}{5}$ &  $\checkmark$\\
    94 & $\boxDims{1}{3}{11}$ & $\boxDims{3}{4}{5}$ &  $\checkmark$\\
    94 & $\boxDims{1}{3}{11}$ & $\boxDims{1}{5}{7}$ &  $\checkmark$\\

    96 & $\boxDims{1}{6}{6}$ & $\boxDims{2}{2}{11}$ &  $\checkmark$\\
    96 & $\boxDims{1}{6}{6}$ & $\boxDims{4}{4}{4}$ &  $\checkmark$\\
    100 & $\boxDims{1}{2}{16}$ & $\boxDims{2}{4}{7}$ &  $\checkmark$\\
    102 & $\boxDims{2}{3}{9}$ & $\boxDims{3}{3}{7}$ &  $\checkmark$\\
    102 & $\boxDims{1}{3}{12}$ & $\boxDims{2}{3}{9}$ &  $\checkmark$\\
    104 & $\boxDims{2}{2}{12}$ & $\boxDims{2}{5}{6}$ &  $\checkmark$\\
    106 & $\boxDims{1}{1}{26}$ & $\boxDims{1}{2}{17}$ &  $\checkmark$\\
    106 & $\boxDims{1}{2}{17}$ & $\boxDims{1}{5}{8}$ &  $\checkmark$\\
        \bottomrule
        \end{tabular}
    \begin{tabular}{@{}c l@{\hspace{6pt}}lc@{}}
    \toprule
    Area\! & \multicolumn{2}{c}{~Dimensions~} & \!First\\
    \midrule

    108 & $\boxDims{1}{4}{10}$ & $\boxDims{3}{4}{6}$ &  $\checkmark$\\
    110 & $\boxDims{1}{3}{13}$ & $\boxDims{1}{6}{7}$ &  $\checkmark$\\
    110 & $\boxDims{1}{3}{13}$ & $\boxDims{3}{5}{5}$ &  $\checkmark$\\
    112 & $\boxDims{2}{2}{13}$ & $\boxDims{2}{4}{8}$ &  $\checkmark$\\
    112 & $\boxDims{2}{3}{10}$ & $\boxDims{4}{4}{5}$ &  $\checkmark$\\
    112 & $\boxDims{1}{2}{18}$ & $\boxDims{4}{4}{5}$ &  $\checkmark$\\
    118 & $\boxDims{1}{4}{11}$ & $\boxDims{1}{5}{9}$ &  $\checkmark$\\
    118 & $\boxDims{1}{3}{14}$ & $\boxDims{1}{5}{9}$ &  $\checkmark$\\
    118 & $\boxDims{1}{4}{11}$ & $\boxDims{2}{5}{7}$ &  $\checkmark$\\
    118 & $\boxDims{1}{3}{14}$ & $\boxDims{2}{5}{7}$ &  $\checkmark$\\
    120 & $\boxDims{2}{2}{14}$ & $\boxDims{2}{6}{6}$ &  $\checkmark$\\
    124 & $\boxDims{1}{6}{8}$ & $\boxDims{2}{4}{9}$ &  $\checkmark$\\
    126 & $\boxDims{1}{3}{15}$ & $\boxDims{3}{5}{6}$ &  $\checkmark$\\
    126 & $\boxDims{1}{7}{7}$ & $\boxDims{3}{5}{6}$ &  $\checkmark$\\
    126 & $\boxDims{1}{7}{7}$ & $\boxDims{3}{3}{9}$ &  $\checkmark$\\
    126 & $\boxDims{3}{3}{9}$ & $\boxDims{3}{5}{6}$ &  $\checkmark$\\
    128 & $\boxDims{1}{4}{12}$ & $\boxDims{4}{4}{6}$ &  $\checkmark$\\
    128 & $\boxDims{2}{2}{15}$ & $\boxDims{4}{4}{6}$ &  $\checkmark$\\
    132 & $\boxDims{2}{3}{12}$ & $\boxDims{2}{5}{8}$ &  $\checkmark$\\
    136 & $\boxDims{2}{4}{10}$ & $\boxDims{2}{6}{7}$ &  $\checkmark$\\
    136 & $\boxDims{2}{6}{7}$ & $\boxDims{3}{4}{8}$ &  $\checkmark$\\
    138 & $\boxDims{1}{4}{13}$ & $\boxDims{1}{6}{9}$ &  $\checkmark$\\
    138 & $\boxDims{1}{6}{9}$ & $\boxDims{3}{3}{10}$ &  $\checkmark$\\
    142 & $\boxDims{1}{5}{11}$ & $\boxDims{1}{7}{8}$ &  $\checkmark$\\
    142 & $\boxDims{2}{3}{13}$ & $\boxDims{3}{5}{7}$ &  $\checkmark$\\
    142 & $\boxDims{1}{5}{11}$ & $\boxDims{2}{3}{13}$ &  $\checkmark$\\
    142 & $\boxDims{1}{7}{8}$ & $\boxDims{2}{3}{13}$ &  $\checkmark$\\
    142 & $\boxDims{1}{5}{11}$ & $\boxDims{3}{5}{7}$ &  $\checkmark$\\
    144 & $\boxDims{2}{2}{17}$ & $\boxDims{3}{6}{6}$ &  $\checkmark$\\
    148 & $\boxDims{1}{4}{14}$ & $\boxDims{4}{5}{6}$ &  $\checkmark$\\
    174 & $\boxDims{3}{3}{13}$ & $\boxDims{3}{5}{9}$ &  $\checkmark$\\
    \bottomrule
\end{tabular}
\end{table}

%% file: sections/related_work.tex
The algorithmic search for common unfoldings of non-isomorphic boxes has been investigated in many earlier works \cite{22allsols2011,randomAlg2008,threeBoxes2012,tadaki2020searchdevelopmentsboxhaving,ueharaSurvey2015,XU20171} using various techniques. One body of work is the complete enumeration of all unfoldings \cite{22allsols2011,XU20171}, computing all possible unfoldings between two boxes. The best previous result was due to \cite{XU20171}, providing a complete enumeration of all unfoldings at area 30. Our efficient approach allows us to completely enumerate all unfoldings up to surface area 58, providing complete enumerations for many pairs of boxes that were previously unknown (\Cref{tab:all_sols_times}). Furthermore, solutions of area $22$ were enumerated by \cite{22allsols2011} taking 10 hours, whereas our approach takes 2 minutes. Solutions of area $30$ were enumerated by \cite{XU20171} taking 10 days using a computer with 128 GB memory, our approach only takes 10 minutes using a standard laptop.

We were also able to find many new common unfoldings between boxes that were previously unknown. The best lower-bound for $\minImpossible(2)$ from earlier works was $40$ and already at area $42$ it was not known if boxes of dimensions $\boxDims{1}{1}{10},\boxDims{2}{3}{3}$ admit a common unfolding. Using our approach, we compute all possible common unfoldings between all boxes up to surface area $86$, more than doubling the previous best bound.

The use of SAT solvers to find common unfoldings has also been explored in earlier work \cite{tadaki2020searchdevelopmentsboxhaving}, our approach has two fundamental improvements over the previous approach. The first being the use of implicit equivalences which alleviates the need to consider nets as subsets of the 2D plane (\Cref{sec:common_unfoldings}). This significantly improves the efficiency of the encodings by eliminating the auxiliary variables needed to represent the nets in 2D. Secondly, the use of local constraints (\Cref{sec:local}) allows us to efficiently detect connectedness in the cut edges of boxes, and furthermore this does not place an \emph{a priori} restriction on the radius of the underlying net. Thereby allowing us to completely enumerate \emph{all solutions} between boxes and establishing $\minPossible(3) > 58$ rather than only finding solutions with a bounded radius. 
One can theoretically enumerate all solutions by using the surface area as the upper-bound, but this is too costly on performance \cite{tadaki2020searchdevelopmentsboxhaving}.

%% file: sections/conclusion.tex
We have presented a new SAT-based approach for finding and enumerating polyominos that can be folded into multiple non-isomorphic boxes, outperforming previous approaches. 
We have introduced a technique that could be applicable to searching for other combinatorial objects: approximating constraints that are hard to encode, or result in large numbers of clauses, with simpler local constraints, and then discard the solutions that do not satisfy the original constraints. 

Further supporting the correctness of our encoding, we were able to reproduce the common unfolding of $3$ boxes at area $532$ \cite{threeBoxes2012} as shown in \Cref{fig:532_unfold}. This was obtained by specifying the cut edges for one of the boxes ($\boxDims{2}{13}{16}$) with unit clauses and using a SAT solver to solve for the remaining variables.

The quest for the smallest area allowing a common unfolding of three boxes remains open, and part of our future work. Another promising direction is to formally verify our work, in the line of~\cite{subercaseaux_et_al:LIPIcs.ITP.2024.35}, given how intricate the encoding is and the delicate arguments for its completeness. We suspect that a particularly challenging aspect of that verification will be to formally define \emph{(un)foldings} (see Demaine and O'Rourke~\cite[Ch. 15]{demaine2007geometric}).

\begin{figure}[h!]
    \begin{subfigure}{0.31\textwidth}
        \centering
        \switch{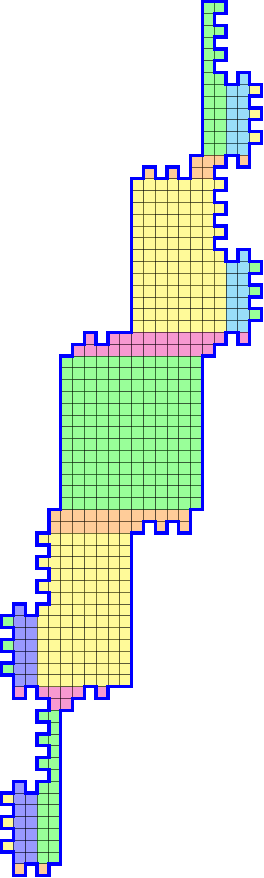}
            \caption{$\boxDims{2}{13}{16}$}
    \end{subfigure}
    \begin{subfigure}{0.31\textwidth}
        \centering
        \switch{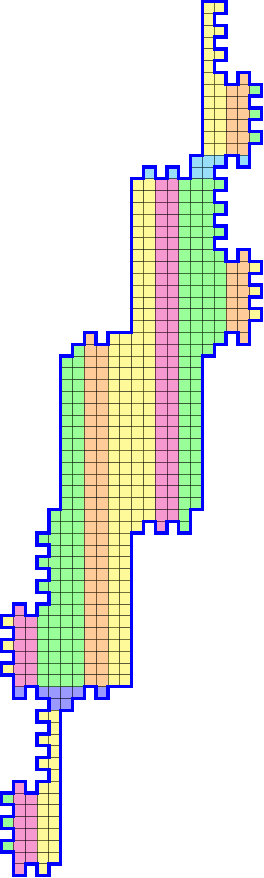}
            \caption{$\boxDims{2}{4}{43}$}  
    \end{subfigure}
    \begin{subfigure}{0.31\textwidth}
        \centering
        \switch{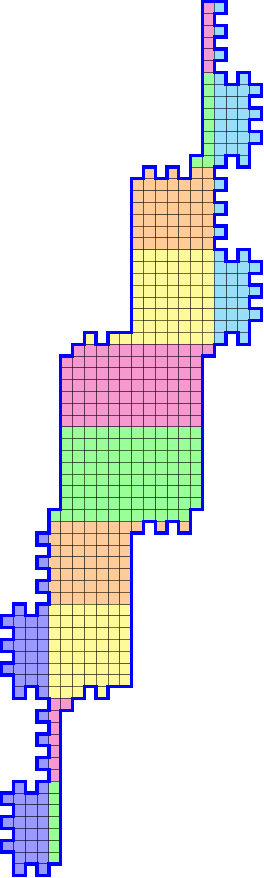}
            \caption{$\boxDims{7}{8}{14}$} 
    \end{subfigure}
    \caption{Common unfolding of 3 boxes at area 532 \cite{threeBoxes2012}.}\label{fig:532_unfold}
\end{figure}

\section*{Acknowledgements}
We thank Joseph Reeves for implementing solution enumeration in CaDiCaL, which was crucial to
performing the experiments. 
This work was supported by the U. S. National Science Foundation under grant DMS-2434625.

%% file: sections/formalization.tex
This section presents the proofs missing from the main text, as well as additional results that were omitted due to space constraints.
First,~\Cref{subsec:1-sink} presents theoretical support for why the edge directions make disconnected solutions more unlikely.
Then,~\Cref{subsec:local-constraints} proves the correctness of the local constraints used in our SAT encoding.
Finally,~\Cref{sec:extra_unfolding} presents additional unfoldings for two boxes.

\subsection{1-Sink Argument}\label{subsec:1-sink}

Let us define $\mathcal{G}$ as the undirected graph whose vertex set is $\mathbb{N}^2$ and whose edge set is $\{ \{u, v\} \in \mathbb{N}^2  : \|u-v\|_1 = 1\}$.
We define an \emph{oriented polyomino} as a directed graph that can be obtained by orienting a finite induced subgraph of $\mathcal{G}$. An example is depicted in \Cref{fig:oriented-polyomino}. 

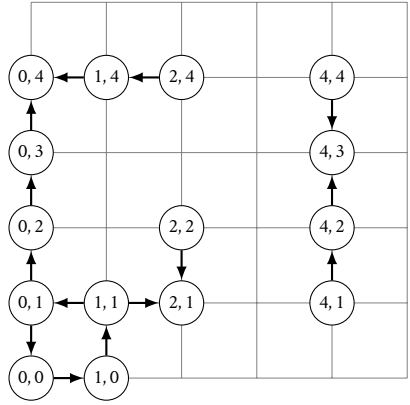
\begin{figure}[h]
    \centering
    \begin{tikzpicture}[scale=1,
        vertex/.style={circle,draw,fill=white,inner sep=2pt},
        every node/.style={font=\scriptsize}
    ]
        \draw[step=1cm,gray,very thin] (0,0) grid (5, 5);
    
        \node[vertex] (A) at (0,0) {$0,0$};
        \node[vertex] (B) at (1,0) {$1,0$};
        \node[vertex] (C) at (1,1) {$1,1$};
        \node[vertex] (D) at (0,1) {$0,1$};
        \node[vertex] (E) at (2,1) {$2,1$};
        \node[vertex] (F) at (2,2) {$2,2$};

        \node[vertex] (Q) at (4,4) {$4,4$};
        \node[vertex] (W) at (4,3) {$4,3$};
        \node[vertex] (R) at (4,2) {$4,2$};
        \node[vertex] (T) at (4,1) {$4,1$};

        \node[vertex] (x1) at (0, 2) {$0,2$};
        \node[vertex] (x2) at (0, 3) {$0,3$};
        \node[vertex] (x3) at (0, 4) {$0,4$};
        \node[vertex] (x4) at (1, 4) {$1,4$};
        \node[vertex] (x5) at (2, 4) {$2,4$};

        \draw[-latex,thick] (x1) -- (x2);
        \draw[-latex,thick] (x5) -- (x4);
        \draw[-latex,thick] (x4) -- (x3);
        \draw[-latex,thick] (x2) -- (x3);
        \draw[-latex,thick] (D) -- (x1);
        \draw[-latex,thick] (A) -- (B);
        \draw[-latex,thick] (B) -- (C);
        \draw[-latex,thick] (C) -- (D);
        \draw[-latex,thick] (D) -- (A);
        \draw[-latex,thick] (C) -- (E);
        \draw[-latex,thick] (F) -- (E);

        \draw[-latex,thick] (Q) -- (W);


        \draw[-latex,thick] (R) -- (W);
        \draw[-latex,thick] (T) -- (R);

    \end{tikzpicture}
    \caption{Example of an oriented polyominio.}\label{fig:oriented-polyomino}
\end{figure}

The \emph{hole number} of an oriented polyomino $P$, denoted $h(P)$, corresponds intuitively to the number of holes in the polyomino. Formally, $h(P)$ is the number of connected components of the subgraph of $\mathcal{G}$ induced by $V(\mathcal{G}) \setminus V(P)$ minus one.

We now define a set $\mathcal{F}$ of 18 \emph{``forbidden subgraphs''}, illustrated in \Cref{fig:forbidden-subgraphs}. An oriented polyomino $P$ is \emph{$\mathcal{F}$-avoidant} if none of its induced subgraphs are isomorphic to a graph in $\mathcal{F}$.

\begin{figure}[h]
    \centering
    \begin{subfigure}{0.18\textwidth}
        \centering
        \begin{tikzpicture}[scale=0.75]
            \node[circle, fill=black, inner sep=3pt] (A) at (0,0) {};
            \node[circle, fill=black, inner sep=3pt] (B) at (1,0) {};
            \node[circle, fill=black, inner sep=3pt] (C) at (2,0) {};
            \draw[-latex,thick] (B) -- (A);
            \draw[-latex,thick] (B) -- (C); 
        \end{tikzpicture}\label{fig:forbidden-subgraphs-a}
        \caption{}
    \end{subfigure}
    \begin{subfigure}{0.12\textwidth}
        \centering
        \begin{tikzpicture}[scale=0.75]
            \node[circle, fill=black, inner sep=3pt] (A) at (0,-1) {};
            \node[circle, fill=black, inner sep=3pt] (B) at (0,0) {};
            \node[circle, fill=black, inner sep=3pt] (C) at (0,1) {};
            \draw[-latex,thick] (B) -- (A);
            \draw[-latex,thick] (B) -- (C); 
        \end{tikzpicture}
        \caption{}\label{fig:forbidden-subgraphs-b}
    \end{subfigure}
    \begin{subfigure}{0.15\textwidth}
        \centering
        \begin{tikzpicture}[scale=0.75]
            \node[circle, fill=black, inner sep=3pt] (A) at (0,0) {};
            \node[circle, fill=black, inner sep=3pt] (B) at (1,0) {};
            \node[circle, fill=black, inner sep=3pt] (C) at (0,1) {};
            \node[circle, fill=black, inner sep=3pt] (D) at (1,1) {};
            \draw[-latex,thick] (B) -- (A);
            \draw[-,thick, blue] (B) -- (D);
            \draw[-,thick, blue] (A) -- (C); 
            \draw[-latex,thick] (C) -- (D); 
        \end{tikzpicture}
        \caption{}\label{fig:forbidden-subgraphs-c}
    \end{subfigure}
    \begin{subfigure}{0.15\textwidth}
        \centering
        \begin{tikzpicture}[scale=0.75]
            \node[circle, fill=black, inner sep=3pt] (A) at (0,0) {};
            \node[circle, fill=black, inner sep=3pt] (B) at (1,0) {};
            \node[circle, fill=black, inner sep=3pt] (C) at (0,1) {};
            \node[circle, fill=black, inner sep=3pt] (D) at (1,1) {};
            \draw[-latex,thick] (A) -- (B);
            \draw[-,thick, blue] (D) -- (B);
            \draw[-,thick, blue] (C) -- (A); 
            \draw[-latex,thick] (D) -- (C); 
        \end{tikzpicture}
        \caption{}\label{fig:forbidden-subgraphs-d}
    \end{subfigure}
        \begin{subfigure}{0.15\textwidth}
            \centering
            \begin{tikzpicture}[scale=0.75]
                \node[circle, fill=black, inner sep=3pt] (A) at (0,0) {};
                \node[circle, fill=black, inner sep=3pt] (B) at (1,0) {};
                \node[circle, fill=black, inner sep=3pt] (C) at (0,1) {};
                \node[circle, fill=black, inner sep=3pt] (D) at (1,1) {};
                \draw[-,thick, blue] (A) -- (B);
                \draw[-latex,thick] (D) -- (B);
                \draw[-latex,thick] (A) -- (C); 
                \draw[-,thick, blue] (D) -- (C); 
            \end{tikzpicture}
            \caption{}\label{fig:forbidden-subgraphs-e}
        \end{subfigure}
            \begin{subfigure}{0.15\textwidth}
                \centering
                \begin{tikzpicture}[scale=0.75]
                    \node[circle, fill=black, inner sep=3pt] (A) at (0,0) {};
                    \node[circle, fill=black, inner sep=3pt] (B) at (1,0) {};
                    \node[circle, fill=black, inner sep=3pt] (C) at (0,1) {};
                    \node[circle, fill=black, inner sep=3pt] (D) at (1,1) {};
                    \draw[-,thick, blue] (A) -- (B);
                    \draw[-latex,thick] (B) -- (D);
                    \draw[-latex,thick] (C) -- (A); 
                    \draw[-,thick, blue] (D) -- (C); 
                \end{tikzpicture}
            \caption{}\label{fig:forbidden-subgraphs-f}
         \end{subfigure}
        
         \vspace{0.7em}
  
         \begin{subfigure}{0.15\textwidth}
          \centering
          \begin{tikzpicture}[scale=0.75]
              \node[circle, fill=black, inner sep=3pt] (A) at (0,0) {};
              \node[circle, fill=black, inner sep=3pt] (B) at (1,0) {};
              \node[circle, fill=black, inner sep=3pt] (C) at (0,1) {};
              \node[circle, draw, inner sep=3pt] (D) at (1,1) {};
              \draw[-latex,thick] (A) -- (B);
              \draw[-latex,thick] (A) -- (C); 
          \end{tikzpicture}
      \caption{}\label{fig:forbidden-subgraphs-g}
      \end{subfigure}
      \begin{subfigure}{0.15\textwidth}
        \centering
        \begin{tikzpicture}[scale=0.75]
            \node[circle, fill=black, inner sep=3pt] (A) at (0,0) {};
            \node[circle, draw, inner sep=3pt] (B) at (1,0) {};
            \node[circle, fill=black, inner sep=3pt] (C) at (0,1) {};
            \node[circle, fill=black, inner sep=3pt] (D) at (1,1) {};
            \draw[-latex,thick] (C) -- (A); 
            \draw[-latex,thick] (C) -- (D); 
        \end{tikzpicture}
    \caption{}\label{fig:forbidden-subgraphs-h}
    \end{subfigure}
    \begin{subfigure}{0.15\textwidth}
      \centering
      \begin{tikzpicture}[scale=0.75]
          \node[circle, draw, inner sep=3pt] (A) at (0,0) {};
          \node[circle, fill=black, inner sep=3pt] (B) at (1,0) {};
          \node[circle, fill=black, inner sep=3pt] (C) at (0,1) {};
          \node[circle, fill=black, inner sep=3pt] (D) at (1,1) {};
          \draw[-latex,thick] (D) -- (B);
          \draw[-latex,thick] (D) -- (C); 
      \end{tikzpicture}
  \caption{}\label{fig:forbidden-subgraphs-i}
  \end{subfigure}
  \begin{subfigure}{0.15\textwidth}
    \centering
    \begin{tikzpicture}[scale=0.75]
        \node[circle, fill=black, inner sep=3pt] (A) at (0,0) {};
        \node[circle, fill=black, inner sep=3pt] (B) at (1,0) {};
        \node[circle, draw, inner sep=3pt] (C) at (0,1) {};
        \node[circle, fill=black, inner sep=3pt] (D) at (1,1) {};
        \draw[-latex,thick] (B) -- (A);
        \draw[-latex,thick] (B) -- (D);
    \end{tikzpicture}
  \caption{}\label{fig:forbidden-subgraphs-j}
  \end{subfigure}
    \caption{The set $\mathcal{F}$ of forbidden subgraphs. The blue edges represent that the graph belongs to $\mathcal{F}$ regardless of the orientation of its blue edges.
    For the forbidden patterns in the second row to match a subgraph, the white node must not be present in the graph.
    There is a total of
    $2 + 4 \cdot 4 - 4 = 14$ (using inclusion-exclusion) forbidden subgraphs in the first row, and thus a total of 18 forbidden patterns.}\label{fig:forbidden-subgraphs}
  \end{figure}
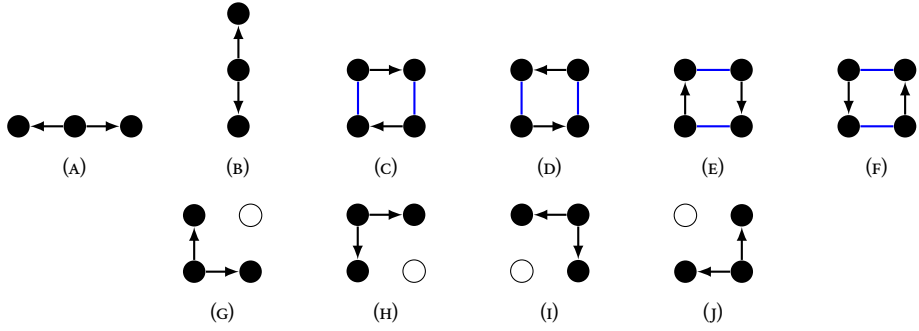

We will use the following lemma, which is a standard consequence of Jordan's curve theorem.
\begin{lemma}\label{lemma:even-odd}
Let $A \subseteq \mathbb{R}^2$ be a simple polygon, and $p$ a point in $\mathbb{R}^2$ that is not on the boundary of $A$. Let $v \neq 0$ be any vector in $\mathbb{R}^2$. Then, $p$ is in the interior of $A$ if and only if the ray $\{p + t \cdot v \mid t \geq 0\}$ intersects the boundary of $A$ an odd number of times.
\end{lemma}

Let us now state the main theorem of this subsection.
\begin{theorem}\label{thm:1-sink}
    Let $P$ be an $\mathcal{F}$-avoidant oriented polyomino containing a directed cycle. Then $h(P) > 0$.
\end{theorem}

\begin{proof}
    Let $P$ be an $\mathcal{F}$-avoidant oriented polyomino containing a directed cycle $C = v_1 \to v_2 \to \cdots \to v_m \to v_1$. Let $X = \{ x \in \mathbb{N} \mid (x, y) \in V(C) \text{ for some } y \in \mathbb{N}\}$, and $Y = \{ y \in \mathbb{N} \mid (x, y) \in V(C) \text{ for some } x \in \mathbb{N}\}$; that is, the $x$- and $y$-coordinates of the vertices in $C$. Observe now that both $|X| > 1$ and $|Y| > 1$, since otherwise there would be a vertex $v \in V(C)$ with total degree $1$ in $C$, but every vertex in a directed cycle has total degree $2$.
    Therefore, if we write $m := \min X$ and $M := \max X$, we have $m < M$, and there must be two distinct vertices $\ell, r$ in $V(C)$ such that $\ell = (m, y_\ell)$ and $r = (M, y_r)$ for some $y_\ell, y_r \in \mathbb{N}$.
    Now let $\pi$ be the path in $C$ from $\ell$ to $r$, and observe that since $\pi$ goes to the right overall, at some point in $\pi$ there must be an edge going to the right for the first time, arriving to a vertex $v := (m + 1, y_v)$.
    Conversely, let $\pi'$ be the path in $C$ from $r$ to $\ell$, and observe that since $\pi'$ starts at $x$-coordinate $M$, and reaches $\ell$ at $x$-coordinate $m$, then at some point in $\pi'$ there is at lease one edge going left and arriving to a vertex $v' := (m, y_{v'})$, with $y_v \neq y_{v'}$. Among these, take the $v'$ that minimizes $|y_v - y_{v'}|$.
   
    We now claim that all edges in the path from $\ell = (m,  y_\ell)$ to $(m, y_v)$ go in the same vertical direction, staying at $x$-coordinate $m$. Indeed, if there was an edge going left, that would contradict the minimality of $m$, and if there was an edge going right, that would contradict the fact that $v$ is the first vertex in $\pi$ with an incoming right edge.
    Similarly, all edges in the path from $(m, y_{v'})$ to $\ell$ go in the same vertical direction, the opposite of the previous one, for otherwise we would either break the minimality of $m$, or the minimality of $|y_v - y_{v'}|$.
    As a result, all points $(m, y)$ with $y \in \{\min(y_v, y_{v'}), \min(y_v, y_{v'}) + 1, \ldots, \max(y_v, y_{v'}) \}$ belong to $V(C)$, and thus to $V(P)$.

    Consider now the set of ``parallel'' undirected edges  
    \[ 
    S := \{ \{(m, y), (m+1, y)\} \mid y \in \{\min(y_v, y_{v'}), \min(y_v, y_{v'}) + 1, \ldots, \max(y_v, y_{v'}) \} \},
    \]
    which is depicted with dashed red lines in \Cref{fig:1-sink}. 

        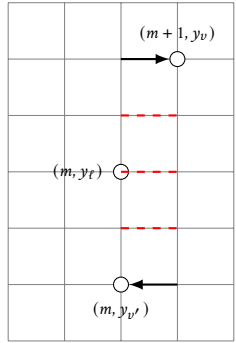
\begin{figure}
            \centering
            \begin{tikzpicture}[    
                scale = 0.75,   
                vertex/.style={circle,draw,fill=white,inner sep=2pt},
                every node/.style={font=\tiny}]
                \draw[step=1cm,gray,very thin] (0,0) grid (4, 6);

                \node[vertex, label={left:$(m, y_\ell)$}] (A) at (2, 3) {};

                \node[vertex, label={$(m+1, y_v)$}] (B) at (3, 5) {};

                \node[vertex, label={below:$(m, y_{v'})$}] (C) at (2, 1) {};

                \draw[-latex,thick] (3, 1) -- (C);

                \draw[-latex,thick] (2, 5) -- (B);

                \draw[-,thick, dashed, red] (2, 4) -- (3, 4);

                \draw[-,thick, dashed, red] (2, 3) -- (3, 3);

                \draw[-,thick, dashed, red] (2, 2) -- (3, 2);

            \end{tikzpicture}
            \caption{Illustration of the proof of \Cref{thm:1-sink}.}\label{fig:1-sink}
        \end{figure}
 
        Let $S_L := \{ (m, y)\mid  y \in \{\min(y_v, y_{v'}), \min(y_v, y_{v'}) + 1, \ldots, \max(y_v, y_{v'}) \}\}$ and $S_R := \{ (m+1, y)\mid  y \in \{\min(y_v, y_{v'}), \min(y_v, y_{v'}) + 1, \ldots, \max(y_v, y_{v'}) \}\}$ be the left and right endpoints of the edges in $S$, respectively.
        Our previous observation was that $S_L \subseteq V(P)$, and also we know that the topmost and bottommost points in $S_R$ are in $V(P)$, since they correspond to $(m+1, y_v)$ and $(m+1, y_{v'})$.
        We now claim that at least one of the points in $S_R$ is not in $V(P)$. For a contradiction, suppose that all points in $S_R$ are in $V(P)$. Then, all edges in $S$ are in $E(P)$ with some orientation. However, the $\mathcal{F}$-avoidant condition implies, by considering the forbidden subgraphs in \Cref{fig:forbidden-subgraphs-c,fig:forbidden-subgraphs-d}, that there cannot be a change of orientation (from left to right, or right to left) between consecutive parallel horizontal edges in $S$. 
        This contradicts the fact that the topmost and bottommost edges in $S$ have opposite orientations in $P$, requiring at least one change along the way. 

        As a result, we have that there is at least one point $q \in S_R$ that is not in $V(P)$. Now, consider the ray $\vec{\rho} := \{(m+1, y_q) + (-1, 0) \cdot t \mid t \geq 0 \}$ going left from $q$, and the simple polygon $A$ induced by the rectilinear embedding of the directed cycle $C$ and its interior. We can now observe that $(m, y_q)$ is the only boundary point of $A$ intersected by $\vec{\rho}$, and thus by~\Cref{lemma:even-odd} we have that $q$ is in the interior of $A$.
        
        Note that there are infinitely many points in $V(\mathcal{G}) \setminus V(P)$, and infinitely many of them must be in the exterior of $A$. Take $w$ to be one of them, and assume for a contradiction that there is a path $\gamma$ between $q$ and $w$ in the subgraph of $\mathcal{G}$ induced by $V(\mathcal{G}) \setminus V(P)$. Then, let $\gamma'$ be the natural rectilinear embedding of $\gamma$ in the plane, and observe that as $\gamma'$ goes from a point in the interior of $A$ to a point in the exterior of $A$, by Jordan's curve theorem, it must intersect the boundary of $A$. As the boundary of $A$ is by definition a union of unit-length axis-aligned segments between lattice points, and so is $\gamma'$, their non-empty intersection must include a lattice point $\omega$. Since each unit-length segment $s$ composing $A$ contains only two lattice points, namely its two endpoints $u_1, u_2$, and by definition of $A$ that implies $u_1, u_2 \in V(C) \subseteq V(P)$, we have that $\omega \in V(P)$. Similarly, since $\gamma'$ is a union of unit-length axis-aligned segments between lattice points, and $\omega$ is a lattice point in $\gamma'$, we have that $\omega \in V(\gamma)$. As a result, we have that $V(\gamma) \cap V(P) \neq \emptyset$, which contradicts the fact that $\gamma$ is a path in the subgraph of $\mathcal{G}$ induced by $V(\mathcal{G}) \setminus V(P)$. We conclude that there is no path between $q$ and $w$ in the subgraph of $\mathcal{G}$ induced by $V(\mathcal{G}) \setminus V(P)$, and thus $h(P) > 0$.


    
\end{proof}

\subsection{Correctness of local constraints}\label{subsec:local-constraints}

\begin{lemma}
Let $P$ be a connected polyomino and $s \in P$ an arbitrary vertex. If there exists a vertex $v = (x_v, y_v) \in V(P)$ such that $v_l = (x_v - 1, y_v), v_r = (x_v + 1, y_v) \in V(P)$ and $d(v, s) > \max(d(v_l, s), d(v_r, s))$, then $h(P) > 0$.
\label{lemma:collinear-forbiden-pattern}
\end{lemma}
\begin{proof}
Let $P_l, P_r$ denote the paths from $s$ to $v_l$ and $v_r$ respectively. Note that if $w \in V(P_l) \cap V(P_r)$ for any $w \in V(P)$, then the position of $w$ in $P_l, P_r$ must necessarily be the same, or else it would contradict $d(v, s) > \max(d(v_l, s), d(v_r, s))$. As such, pick $w \in V(P_l) \cap V(P_r)$ such that $w$ has the largest index in the paths, and let $Q_l, Q_r$ denote the paths from $w$ to $v_l$ and $v_r$ respectively. Because $w$ lies in the intersection of the two paths, we necessarily have $d(v, w) > \max(d(u_l, w), d(u_r, w))$. Consider the directed cycle $C = v \to v_r \to Q_r \to Q_l \to v_l \to v$ which induces a closed loop without self-crossings via its rectilinear embedding, further denote the induced simple polygon from this embedding as $A$. By the Jordan Curve Theorem and the fact that the vectors $(0, \pm 1)$ are \emph{not tangent} to the boundary of $A$ at $v = (x_v, y_v)$, one of the rays $\vec{\rho_{\pm}} = \{(x_v, y_v) + (0, \pm 1) \cdot t \mid t \geq 0 \}$ must point towards the interior of $A$, suppose without loss of generality that $\vec{\rho_{-}}$ does. As the boundary of $A$ is formed by unit-length axis-aligned segments, there must be some minimal positive integer $k > 0$ such that $(x_v, y_v - k)$ first intersects the boundary of $A$. Consider the sequence of vertices $v = (x_v, y_v), (x_v, y_v - 1), \cdots, (x_v, y_{v - k})$, it follows by construction that $(x_v, y_v), (x_v, y_v - k) \in V(P)$. Suppose for the sake of contradiction that $h(P) = 0$, which necessarily implies $\{(x_v, y_v - i)\}_{1 \leq i \leq k - 1} \subseteq V(P)$ as they all belong to the interior of $A$ and would result in a hole otherwise. Thus, the path $Q = (x_v, y_v) \to (x_v, y_v - 1) \to \cdots \to (x_v, y_v - k)$ is a path in $P$ with $(x_v, y_v - k) \in Q_l \cup Q_r$, assume without loss of generality that $(x_v, y_v - k) \in Q_r$. As $v_r = (x_v + 1, y_v)$ and $P$ is a subgraph of $\mathcal{G}$, this further implies $d(v_r, (x_v, y_v - k)) \geq d(v, (x_v, y_v - k)) + 1$ and therefore $d(v_r, w) \geq d(v, w) + 1$, a contradiction. 
\end{proof}

In what follows, we will consider undirected and connected graphs with no self-loops, and use notation $d(u, v)$ for the distance between two vertices $u, v$ in the graph, defined as the length of the shortest path between $u$ and $v$.
\begin{lemma}
    Let $G = (V, E)$ be an undirected and connected graph, and $\{u, v\} \in E$ for some vertices $u,v$. Then, for any vertex $w \in V$, we have that $|d(u, w) - d(v, w)| \leq 1$.
    \label{lemma:dist-1}
\end{lemma}
\begin{proof}
    The triangle inequality says $d(u, w) \leq d(u, v) + d(v, w) = 1 + d(v, w)$, and symmetrically $d(v, w) \leq 1 + d(u, w)$. 
\end{proof}

\begin{lemma}
    Let $G = (V, E)$ be an undirected, connected, and bipartite graph. Then, for any edge $\{u, v\} \in E$, and vertex $w \in V$, we have that $d(u, w) \neq d(v, w)$.
    \label{lemma:neq-d-bipartite}
\end{lemma}
\begin{proof}
    Suppose for a contradiction that $d(u, w) = d(v, w) = c$. 
    Then, there is a path $P$ from $u$ to $w$ of length $c$, and a path $P'$ from $w$ to $v$ of length $c$. Thus, $v \to P \to P'$ is a cycle of length $1 + 2c$, which contradicts the fact that $G$ is bipartite.
\end{proof}

Combining the previous two lemmas, we get the following.
\begin{lemma}
    Let $G = (V, E)$ be an undirected, connected, and bipartite graph. Let $w \in V$ be any vertex, and $\{u, v\} \in E$ any edge. Then, we have that either $d(u, w) = d(v, w) + 1$ or $d(v, w) = d(u, w) + 1$.
\end{lemma}

Recall now that an orientation of an undirected graph $G$ corresponds to a directed graph $G'$ obtained by replacing each edge $\{u, v\} \in E(G)$ with exactly one of the two directed edges $(u, v)$ or $(v, u)$.

\begin{definition}
    Let $G = (V, E)$ be an undirected, connected, and bipartite graph and $s \in V$ some vertex. We define the \emph{``orientation of $G$ towards $s$''}, denoted as $G_{\downarrow s}$ as the directed graph obtained by orienting each edge $\{u, v\} \in E$ as going from $u$ to $v$ if $d(s, u) = d(s, v) + 1$, and from $v$ to $u$ if $d(s, v) = d(s, u) + 1$. 
    \label{def:oriented-towards}
\end{definition}

Note that this is well-defined by~\Cref{lemma:neq-d-bipartite}. 

\begin{lemma}
    Let $G = (V, E)$ be an undirected graph and $s \in V$ some vertex. Then, $G_{\downarrow s}$ has no directed cycles.
\label{lemma:orientation-towards-no-cycles}
\end{lemma}
\begin{proof}
    Suppose for a contradiction that $G_{\downarrow s}$ has a directed cycle $C = v_1 \to v_2 \to \cdots \to v_m \to v_1$, with $m > 1$. Then, by the observation above, we have that $d(s, v_1) = d(s, v_2) + 1 = d(s, v_3) + 2 = \cdots = d(s, v_m) + (m - 1) = d(s, v_1) + m$, a contradiction since $m > 1$.
\end{proof}

\begin{lemma}
    Consider oriented polyomino $P$ that is connected and oriented towards a vertex $s \in V(P)$. Then, $P$ does not contain any of the patterns depicted in \Cref{fig:forbidden-subgraphs-c,fig:forbidden-subgraphs-d,fig:forbidden-subgraphs-e,fig:forbidden-subgraphs-f}. 
    \label{lemma:no-opposite-parallel-edges}
\end{lemma}
\begin{proof}
    \begin{figure}
        \centering
    \begin{tikzpicture}
        \node[circle, fill=black, inner sep=3pt, label=left:$u$] (A) at (0,0) {};
        \node[circle, fill=black, inner sep=3pt, label=right:$v$] (B) at (1,0) {};
        \node[circle, fill=black, inner sep=3pt, label=left:$w$] (C) at (0,1) {};
        \node[circle, fill=black, inner sep=3pt, label=right:$x$] (D) at (1,1) {};
        \draw[-latex,thick] (B) -- (A);
        \draw[-latex,thick] (D) -- (B);
        \draw[-latex,thick] (C) -- (A); 
        \draw[-latex,thick] (C) -- (D); 
    \end{tikzpicture}
    
    \caption{Forbidden pattern for the proof of~\Cref{lemma:no-opposite-parallel-edges}.}\label{fig:forbidden-pattern}
\end{figure}
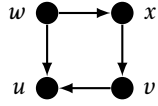
    First, note that any polyomino, when taken as an undirected graph, is bipartite since it is a subgraph of $\mathcal{G}$, which is bipartite.
    Thus, we can apply~\Cref{lemma:orientation-towards-no-cycles} to obtain that $P$ has no directed cycles. Therefore, out of the patterns in~\Cref{fig:forbidden-subgraphs-c,fig:forbidden-subgraphs-d,fig:forbidden-subgraphs-e,fig:forbidden-subgraphs-f}, it suffices to check without loss of generality that it cannot contain the pattern depicted in~\Cref{fig:forbidden-pattern}.
    Indeed, suppose expecting a contradiction that $P$ contains the pattern in~\Cref{fig:forbidden-pattern}. 
    Then, by~\Cref{def:oriented-towards}, we have the following equations: 
    \begin{align*}
        d(s, w) &= d(s, u) + 1, \tag{edge $(w, u)$}\\
        d(s, w) &= d(s, x) + 1, \tag{edge $(w, x)$}\\
        d(s, x) &= d(s, v) + 1, \tag{edge $(x, v)$}\\
        d(s, v) &= d(s, u) + 1. \tag{edge $(v, u)$}\\
    \end{align*}
    The first and last equations together imply $d(s, w) = d(s, v)$, and the second and third equations together imply $d(s, w) = d(s, v) + 2$, a contradiction.       
\end{proof}

The following lemma shows that the second row of patterns in \Cref{fig:forbidden-subgraphs-body} are also avoided by a polyomino without holes. 

\begin{lemma}
    \label{lemma:orthogonal-forbidden-pattern}
    Let $U$ be a connected polyomino without holes, then $P := U_{\downarrow s}$ for any $s \in U$ avoids the second row of \Cref{fig:forbidden-subgraphs}.
\end{lemma}
\begin{proof}
    Since the patterns are equivalent up to rotation, it suffices to show that \Cref{fig:forbidden-subgraphs-j} is avoided. Suppose that there exists vertices $v, v_l, v_u \in U$ where $v = (x_v, y_v), v_l = (x_v - 1, y_v), v_u = (x_v, y_v + 1)$ and $d(v, s) > \max(d(v_l, s), d(v_u, s))$. Let $P_l, P_u$ denote the paths from $s$ to $v_l$, $v_u$ respectively, and assume without loss of generality that $P_l \cap P_u = \{s\}$ (otherwise move the sink to the last intersection). Denote by $A$ the simple polygon formed by the rectilinear embedding of $P_l \cup P_u \cup v$. At $v$, since the vector $(-1, 1)$ is \emph{not parallel} to the boundary of $A$, one of the vectors $(-1, 1), (1, -1)$ must point towards in the interior of $A$. If $(-1, 1)$ does, then it follows that $(x_v - 1, y_v + 1) \in P$ (as $h(P) = 0$), and the pattern is avoided. Now suppose $(1, -1)$ point towards the interior of $A$, which implies $v_r := (x_v + 1, y_v - 1) \in P$. Consequently, since $h(P) = 0$, it follows that the sequence of vertices $v_r = (x_v + 1, y_v - 1), (x_v + 1, y_v), \cdots, (x_v + 1, y_v + k)$ all belong to $P$ until the first $k$ such that $v_f := (x_v + 1, y_v + k) \in P_l \cup P_u$, assume without loss of generality that it belongs to $P_u$. The path from $v_r$ to $(x_v + 1, y_v + k)$ is a geodesic as it only goes up, so this implies $d(v_r, v_f) \leq d(v_u, v_f)$ and consequently $d(v_r, s) \leq d(v_u, s) < d(v, s)$. So the edge $\{v, v_r\}$ will be directed towards $v_r$, contradicting \Cref{lemma:collinear-forbiden-pattern}. 
\end{proof}

\begin{theorem}
    \label{thm:forbidden-patterns}
    Let $U$ be a connected undirected polyomino without holes (i.e., $h(U) = 0$) and $s \in V(U)$ some vertex. Then, $P := U_{\downarrow s}$ is an oriented polyomino that is $\mathcal{F}$-avoidant.
\end{theorem}
\begin{proof}
Follows directly from~\Cref{lemma:collinear-forbiden-pattern,lemma:no-opposite-parallel-edges,lemma:orthogonal-forbidden-pattern}.
\end{proof}
The following lemma captures the only geometric property we use of folding and unfolding. Note that we do not need a precise definition of what an unfolding motion is, just that it is orientation-preserving (in the geometric sense acting on) and acts along grid-lines. 

\begin{lemma}
    \label{lemma:geometric_folding}
    Let $\{\sqa, \sqb\}$ be adjacent squares on a box $B$ and $o : B \to [4]$ denote the orientation assignment induced by some unfolding of $B$ (\Cref{sec:local}). If the edge $\{\sqa, \sqb\}$ is not cut in the unfolding, then $o(\sqb) = o(\sqa) + r_{\sqa}(\sqb)$. 
\end{lemma}
\begin{proof}
    Note that an unfolding is formed by a sequence $\delta_{1}, \cdots, \delta_{k}$ geometric transformations acting on $B$ along the grid-lines. In particular, each $\delta_i$ is \emph{orientation-preserving}, and therefore the relative orientation $r_{\sqa}(\sqb)$ remains invariant after the action of each $\delta_i$, thus also remains invariant after $\delta_{k}$. 
\end{proof}
This is now enough to prove that our local constraints are necessary conditions of unfoldings. 

\begin{theorem}
    \label{thm:net_implies_valid}
    Let $P$ be a polyomino that folds into a box $B$ (i.e. it gives a net of $B$). This naturally induces a bijection $F : P \to B$ sending vertices of the polyominos to squares on $B$ by folding the net into $B$. Let $C_P$ denote the cut edges induced by this mapping:
    \[C_p = \{\{\sqa, \sqb\} \mid \{F^{-1}(\sqa), F^{-1}(\sqb)\} \notin E(P)\}\]
    where $E(P)$ denotes the edge set of $P$. Then there exists a corresponding satisfying assignment of the local constraints (\Cref{sec:local}) using $C_p$ as the cut edges. I.e. the boolean variable $e_{\sqa, \sqb}$ is true if and only if $\{\sqa, \sqb\} \notin C_p$. 
\end{theorem}

\begin{proof}
    We first begin by constructing the variable assignments.
    \begin{enumerate}
        \item Edges variables $e_{\sqa, \sqb}$: This is given by $C_p$ defined in the statement of the theorem.
        \item Orientation variables $o_{s, r}$: Recall that $F : P \to B$ denotes a folding of $P$ into $B$. Since we think of $B$ as labeled (with dots), we can first fold $P$ into $B$, label $P$ with dots in the same way as $B$ was labeled, and then unfold $P$ (now labeled with dots). After this action, every vertex $p \in P$ now has an orientation value associated with the unfolding, we define the orientation of $o(F(p)) \in [4]$ to be this orientation value. Since $F$ is a bijection, this constructs an assignment $o : B \to [4]$. 
        \item Sink ($s^\star$) and edge directions ($d_{\sqa, \sqb}$): Pick an arbitrary square $s^\star \in B$ as the sink, assign the direction variables $d_{\sqa, \sqb}$ to be the orientations inherited from the oriented polyomino $P_{\downarrow s^\star}$.
    \end{enumerate}
    It remains to verify that such assignments satisfy all constraints introduced in \Cref{sec:local}.
    \begin{enumerate}
        \item Orientation constraints: These constraints only enforce that relative orientations are preserved, therefore directly follow from \Cref{lemma:geometric_folding}.
        \item Forbidding patterns (\Cref{fig:forbidden-pattern}): This follows from the construction of $P_{\downarrow s^\star}$ and \Cref{thm:forbidden-patterns}.
    \end{enumerate}
    All local constraints are therefore satisfied, and the proof is complete. 
\end{proof}

\subsection{Common Unfoldings}
\label{subsec:common_unfoldings}

The following theorem establishes that any common unfolding of boxes $B_1, B_2$ also naturally satisfy the constraint $\text{EQUIV}_{B_1, B_2}$ (\Cref{sec:common_unfoldings}).

\begin{theorem}
    \label{thm:common_net_is_valid}
    Let $P$ be a polyomino that folds into boxes $B_1, B_2$. (I.e. it gives a common unfolding), let $(C_1, C_2)$ denote the cut edges of $B_1, B_2$ respectively induced by $P$ as defined in \Cref{thm:net_implies_valid}. Then there exists a satisfying solution of the equivalence constraints (\Cref{sec:common_unfoldings}) with $(C_1, C_2)$ being the cut edges. 
\end{theorem}

\begin{proof}
    Since $P$ folds into both boxes $B_1, B_2$, denote by $F_i : P \to B_i$ for $i \in [2]$ the maps induced by the foldings. It naturally induces a bijection 
    \[M \coloneqq F_1 \circ F_{2}^{-1}: B_2 \to B_1\] 
    It suffices to show that $M$ is a satisfying solution, with $\rho : B_1 \to [4]$ defined via $\rho(\sqa) = o(\boxb{\sqa}) - o(\sqa)$ where $o$ is the orientation assignment given by \Cref{thm:net_implies_valid}. 
    \begin{itemize}
        \item By construction, the constraints
        \begin{align*}
            \sum_{s \in B_1} m_{s, \boxb{s}} &= 1\\
            m_{s, \boxb{s}} \land o_{\boxb{s}, d + r} \land o_{s, d} &\rightarrow \rho_{s, r}
        \end{align*}
        are satisfied since $M$ is a bijection.
        \item $m_{\sqa, \boxb{\sqa}} \land \rho_{\sqa, r} \land e_{\sqa, \sqb} \rightarrow m_{\sqb, \boxb{\sqb}}$:  Recall that such constraints only enforce preservation of 2D-positions. E.g. If $M(\boxb{\sqa}) = \sqa$ for some $\sqa \in B_1, \boxb{\sqa} \in B_2$ where $F_1^{-1}(s_1) = (x, y)$ \emph{and} $\sqb \in B_1$ satisfies $F_1^{-1}(\sqb) = (x, y + 1)$, then the unique $\boxb{\sqb} \in B_2$ with $F_2^{-1}(\boxb{\sqb}) = (x, y + 1)$ must satisfy $M(\boxb{\sqb}) = \sqb$. But this follows directly by construction:
        \begin{align*}
            M(\boxb{\sqb}) = F_1(F_2^{-1}(\boxb{\sqb})) = F_1((x, y + 1)) = F_1(F_1^{-1}(\sqb)) = \sqb
        \end{align*}
        thus such constraints are satisfied by $M$. Furthermore, for the remaining edge preservation constraints we also have:
        \begin{align*}
            e_{\boxb{\sqa}, \boxb{\sqb}} &\iff \{\boxb{\sqa}, \boxb{\sqb}\} \notin C_2 \iff \{F_2^{-1}(\boxb{\sqa}), F_2^{-1}(\boxb{\sqb})\} \in E(P)\\
            &\iff \{F_1(F_2^{-1}(\boxb{\sqa})), F_1(F_2^{-1}(\boxb{\sqb}))\} \notin C_1 \iff \{M(\boxb{\sqa}), M(\boxb{\sqb})\} \notin C_1\\
            & \iff e_{\sqa, \sqb}
        \end{align*}
        Hence all constraints are satisfied, completing the proof. 
    \end{itemize}
\end{proof}

Finally, this allows us to prove $\minPossible(3) > 58$.

\begin{theorem}
    $\minPossible(3) > 58$. 
\end{theorem}
\begin{proof}
    Let $B_1, B_2, B_3$ be 3 boxes with common area. By \Cref{thm:common_net_is_valid}, if a common unfolding existed, there must necessarily be cut edges $(C_1, C_2, C_3)$ such that all local constraints, $\text{EQUIV}_{B_1, B_2}$ and $\text{EQUIV}_{B_1, B_3}$ can be satisfied. In other words, it suffices to show that for every pair of cut edges $(C_1, C_2)$ satisfying $\text{EQUIV}_{B_1, B_2}$, there does not exist a corresponding $C_3$ where $(C_1, C_3)$ can satisfy $\text{EQUIV}_{B_1, B_3}$. But this is precisely what we show by a complete enumeration of all solutions for $B_1, B_2$ (\Cref{tab:all_sols_times}) up to area $58$. Hence, $\minPossible(3) > 58$. 
\end{proof}

%% file: sections/add_unfold.tex
As described in~\Cref{sec:local}, a shortcoming of the approach of Tadaki and Amano~\cite{tadaki2020searchdevelopmentsboxhaving} is that it relies on a heuristic constraining the \emph{diameter} (largest graph distance between two squares on a net) of solutions, thus potentially missing solutions with relatively large diameters.
Through our encoding, which does not rely on the small-diameter assumption, we are able to find solutions with relatively large diameters, as depicted in~\Cref{fig:large_diam_80}. 



%
\begin{figure}[b]
    \switch{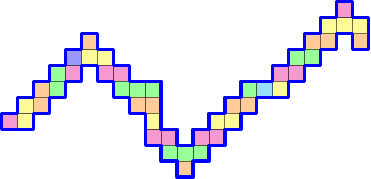}
    \vspace{-20pt}
    
    \begin{subfigure}{0.49\textwidth}
    \switch{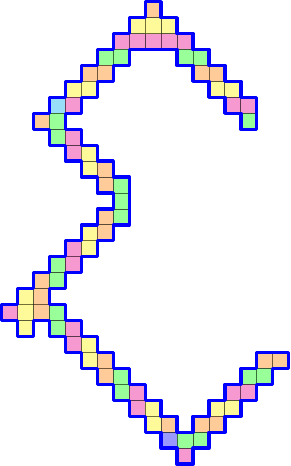}
    \end{subfigure}
    \begin{subfigure}{0.49\textwidth}
    \switch{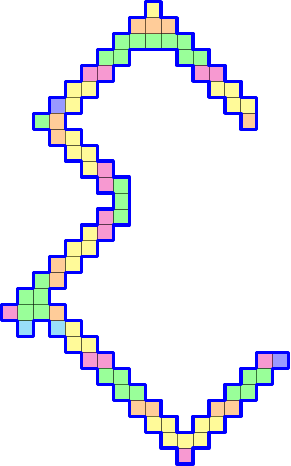}
    \end{subfigure}        
         \vspace{-5pt}
    \caption{Top: A common unfolding of area $46$ for boxes $\boxDims{1}{1}{11}$ (faces shown) and $\boxDims{1}{2}{7}$. The diameter of this net is $41$, the largest among all solutions. Bottom: A common unfolding for boxes $\boxDims{1}{1}{20}$ and $\boxDims{1}{2}{13}$.}\label{fig:large_diam_80}
\end{figure}